\documentclass[preprint]{elsarticle}

\usepackage{natbib}
\usepackage{geometry}
\usepackage{fleqn}
\usepackage{graphicx}
\usepackage{hyperref}
\usepackage{amsmath, amsthm, amssymb, amsfonts,mathrsfs}
\usepackage{braket}
\usepackage{color}
\usepackage{bbold}

\title{Irreducibility and regularisation properties of Gaussian quantum Markov semigroups}
\date{}

\author[1]{Franco Fagnola}
\ead[1]{franco.fagnola@polimi.it}
\affiliation[1]{organization={Dipartimento di Matematica, Politecnico di Milano},
addressline={Via Edoardo Bonardi 9},
postcode={20133},
city={Milano},
country={Italy}}
\cortext[1]{franco.fagnola@polimi.it}

\author[1]{Federico Girotti}
\ead[1]{federico.girotti@polimi.it}
\cortext[1]{federico.girotti@polimi.it}

\newtheorem{defi}{Definition}
\newtheorem{lemma}[defi]{Lemma}
\newtheorem{coro}[defi]{Corollary}
\newtheorem{prop}[defi]{Proposition}
\newtheorem{theo}[defi]{Theorem}
\newtheorem{result}[defi]{Result}

\theoremstyle{definition}

\def\hh{\mathfrak{h} }
\def\nn{\mathbb{N}}
\def\CC{\mathbb{C}}
\def\RR{\mathbb{R}}
\def\FS{\Gamma(\CC^d)}
\def\TT{\mathcal{T}}
\def\LL{\mathcal{L}}
\def\NN{\mathcal{N}(\TT)}
\def\hh{\mathfrak{h}}
\def\tr{{\rm tr}}

\begin{document}
\begin{abstract}
We study regularisation and irreducibility properties of Gaussian quantum Markov semigroups (GQMSs) acting on continuous-variable quantum systems. We first identify a natural notion of regularity for operators in this setting, which allows us to formulate and characterise the smoothing effects of GQMSs in terms of algebraic conditions involving the drift and quantum diffusion matrices. These conditions establish a connection with the controllability theory of quantum linear systems and with the structure of decoherence-free subsystems.

We then characterise irreducibility through several equivalent algebraic criteria: one formulated in terms of the drift and quantum diffusion matrices, one in terms of the operators appearing in the generalised GKLS representation of the generator, and a third given by a quantum analogue of H\"ormander's condition. A central and somewhat surprising consequence is that, in contrast with the classical case, irreducibility is strictly stronger than conditions ensuring regularisation. Our results provide an algebraic framework for analysing these properties and lay the groundwork for a broader study of reducible Gaussian quantum Markov semigroups and, more generally, more general relevant quantum Markov semigroups on continuous-variable systems.
\end{abstract}
\maketitle

\makeatletter
\def\ps@pprintTitle{%
  \let\@oddhead\@empty
  \let\@evenhead\@empty
  \let\@oddfoot\@empty
  \let\@evenfoot\@oddfoot
}
\makeatother

\section{Introduction}

Gaussian quantum Markov semigroups (GQMSs) form a class of evolutions on the space of bounded linear operators on the bosonic Fock space $\FS$ (in this work, we restrict our attention to finitely many modes). They play a central role in the study of open quantum systems, providing convenient models for quantum optical experiments, atomic ensembles, quantum memories, and related applications. For this reason, they have been extensively studied in the field of quantum control under the name of linear quantum systems (see \cite{DZ22,NY17} and references therein).

\bigskip The name Gaussian comes from the fact that GQMSs can be characterised as the class of those quantum Markov semigroups that preserve the set of quantum Gaussian states (\cite{PO22}): given an initial quantum Gaussian state $\rho$ with mean $\mathbf{m} \in \RR^{2d}$ and covariance $\mathbf{\Sigma} \in M_{2d}(\RR)$, its evolution under the semigroup is a family of quantum Gaussian states $(\rho_t)_{t \geq 0}$ with parameters $\mathbf{m}_t$ and $\mathbf{\Sigma}_t$ satisfying the following Cauchy problem:
\begin{align}\label{eq:par}
&\frac{d\mathbf{m}_t}{dt}=\mathbf{Z}^T\mathbf{m}_t+\boldsymbol{\zeta}, \quad \mathbf{m}_0=\mathbf{m},\\
&\frac{d\mathbf{\Sigma}_t}{dt}=\mathbf{Z}^T\mathbf{\Sigma}_t + \mathbf{\Sigma}_t \mathbf{Z} + \mathbf{C}, \quad \mathbf{\Sigma}_0=\mathbf{\Sigma}\nonumber\end{align}
for some $\mathbf{Z}, \mathbf{C} \in M_{2d}(\RR)$ and $\boldsymbol{\zeta} \in \RR^{2d}$ characterising the semigroup. The only constraint that the parameters appearing in the evolution in eq. \eqref{eq:par} need to satisfy is the following one: the matrix with complex entries defined as
\begin{equation} \label{eq:qdiffusion} \mathbf{C}_{Z}:=\mathbf{C}- i(\mathbf{Z}^*\mathbf{J}+\mathbf{J}\mathbf{Z}),\end{equation}
to which we will refer as \textit{quantum diffusion matrix} in the following, needs to be positive semidefinite as a (complex) linear operator acting on $\CC^{2d}$ ($\mathbf{J}$ denotes the symplectic matrix); this ensures that the Heisenberg uncertainty principle is not violated. We remark that condition in eq. \eqref{eq:qdiffusion} can be equivalently reformulated in terms of $\mathbf{C}+i(\mathbf{Z}^*\mathbf{J}+\mathbf{J}\mathbf{Z})$ and every statement that we will prove and that involves $\mathbf{C}- i(\mathbf{Z}^*\mathbf{J}+\mathbf{J}\mathbf{Z})$ keeps holding true replacing it with $\mathbf{C}+ i(\mathbf{Z}^*\mathbf{J}+\mathbf{J}\mathbf{Z})$.

Moreover, the infinitesimal generator of a GQMS can be formally written in a generalised Gorini-Kossakowski-Sudarshan-Lindblad (GKSL) form with a quadratic Hamiltonian $H$ and linear jump operators $\{L_\ell\}_{\ell=1}^m$, where quadratic and linear are to be understood in terms of creation and annihilation operators.

From a mathematical point of view, GQMSs are the noncommutative counterpart of Ornstein-Uhlenbeck semigroups (see for instance \cite{LMP20,FPSU24} and references therein), which is the family of classical Markov semigroup corresponding to the solution of the following stochastic differential equation:
\begin{equation} \label{eq:ito} dX_t=(\mathbf{Z}^TX_t+\boldsymbol{\zeta})dt + \mathbf{B}dW_t,\end{equation}
where $\mathbf{Z} \in M_d(\RR), \, \mathbf{B} \in M_{d\times m}(\RR), \, \boldsymbol{\zeta} \in \RR^d$ and $(W_t)_{t \geq 0}$ is a standard $m$-dimensional Brownian motion. Indeed, if $X_0$ is a Gaussian random variable with mean $\mathbf{m}$ and covariance matrix $\boldsymbol{\Sigma}$, then $X_t$ is a Gaussian random variable for every $t \geq 0$ with mean and covariance evolving according to the Cauchy problem in eq. \eqref{eq:par}, where $\mathbf{C}=\mathbf{B}\mathbf{B}^T.$ From eq. \eqref{eq:ito} one can easily understand why $\mathbf{Z}$ and $\mathbf{C}$ are usually referred to as drift and diffusion matrices, respectively. The infinitesimal generator of the semigroup corresponding to $(X_t)_{t \geq 0}$, which we will denote by $(T_t)_{t \geq 0}$, has the following expression on smooth functions:
\begin{equation}\label{eq:OUfp}
L(f)(x)=\langle \nabla f(x),\mathbf{Z}^Tx+\boldsymbol{\zeta}  \rangle+\frac{1}{2}\tr(\mathbf{C}{\rm Hess}f(x)),\end{equation}
where $\nabla f$ and ${\rm Hess} f$ denote the gradient and the Hessian matrix of $f$, respectively.

\bigskip Ornstein-Uhlenbeck processes have been scrupulously investigated and several important features have been well understood and characterised in terms of the drift and the diffusion matrices (\cite{LMP20}), e.g. existence and uniqueness of stationary measures, irreducibility, regularisation properties, the structure of the spectrum, contractivity properties, etc... On the other hand, despite been ubiquitous in physical applications, a similar rigorous mathematical analysis of GQMSs was still lacking until very recent years (except for some pioneering works, see for instance \cite{CFL00,KP04,CS08}), when there have been considerable progresses in this regard: in \cite{AFP21} existence and uniqueness of normal invariant states and irreducibility were completely understood in one mode; the authors of \cite{AFP22} characterised the decoherence-free subalgebra of GQMSs; hypercontractivity and spectral properties when studying the extension of the semigroup acting on square integrable operators with respect to the unique invariant measure have been investigated in \cite{FPSU24,FL25,Li25,SXZ26}; the existence and the structure of normal invariant states together with the long time behaviour when there exists at least a normal invariant state have been studied in \cite{GP26}; finally, GNS symmetric and weak coupling limit GQMS were completely characterised in \cite{BSFQ26} and \cite{PSU26}, respectively. 

Despite the problem has been attacked in several works (\cite{AFP22b,FP22,GP26}), a complete characterisation of irreducibility of the semigroup in terms of $\mathbf{Z}$, $\mathbf{C}$ and $\boldsymbol{\zeta}$ has remained out of reach so far. It has been shown (\cite{AFP21}) that in one mode irreducibility is equivalent to the condition (that we will denote as \textit{$\mathbb{G}$-cond}) that the complex linear span of the following operators
$$\mathbf{1}, \,L_\ell, \, [G,L_\ell], \, [G, [G,L_\ell]], \dots
$$
contains all creation and annihilation operators, where $G$ is defined as $-iH-\frac{1}{2}\sum_{\ell=1}^{m}L^*_\ell L_\ell$ (we postpone the discussion about domain issues later in the article); this has been proved to be equivalent to irreducibility in many modes only at the cost of assuming further technical assumptions (see \cite{FP22}). With a completely different approach, \cite{GP26} showed that if the semigroup admits a normal invariant state, then irreducibility is equivalent to the following condition (which we will call \textit{$\mathbf{C}_{Z}$-cond}):
$$\text{there are no nontrivial $\mathbf{Z}$-invariant subspaces of } \CC^{2d} \text{ contained in } \ker(\mathbf{C}_{Z}).$$
The relationship between the two condition has remained an open question so far.

\bigskip In order to better understand what requirements on the parameters one should expect to be equivalent to irreducibility, it is enlightening to look at the classical case. For the classical Ornstein-Uhlenbeck semigroup one has that
\begin{equation} \label{eq:classcond}
\text{there are no nontrivial $\mathbf{Z}$-invariant subspaces of $\RR^d$ contained in $\ker(\mathbf{C})$}
\end{equation}
if and only if the corresponding semigroup is $\mathscr{L}$-irreducible, where $\mathscr{L}$ stands for the $d$-dimensional Lebesgue measure (see for instance \cite{MTG09} for the notion of irreducibility with respect to a measure). Indeed, one can check that if $X_0=x \in \RR^d$, then $X_t$ has a Gaussian distribution with strictly positive (under condition in eq. \eqref{eq:classcond}) covariance matrix $\boldsymbol{\Sigma}_t$; therefore, for every $x \in \RR^d$, for every Borel set $A \subseteq \RR^{d}$ such that $\mathscr{L}(A)>0$ and for every strictly positive time $t$, one has that
$\mathbb{P}(X_t \in A|X_0=x)>0.$
Moreover, if condition in eq. \eqref{eq:classcond} holds true, the semigroup has also some regularising effects: for every probability density $f$ with respect to the Lebesgue measure, one has that
\begin{equation}\label{eq:clrepresentation}T_{t*}(f)=g_t *S_t(f),\end{equation}

where $S_t(f)(x)=e^{-t\tr(\mathbf{Z})}f(e^{-t\mathbf{Z}^T}x)$ and $g_t$ is the density function of a Gaussian distribution with covariance matrix $\boldsymbol{\Sigma}_t$. We use the convention of denoting by $T_{t*}$ the semigroup at time $t$ acting on probability densities, that is the one corresponding to the Fokker-Plank equation and the describing the evolution of the law of the underlying stochastic process. Therefore, $T_{t*}(f)$ is infinitely differentiable for every positive time $t$. The condition in eq. \eqref{eq:classcond} can be seen to be equivalent to H\"ormander condition for the simple stochastic differential equation in eq. \eqref{eq:ito}, i.e.
\begin{center} 
at every point in $\RR^{d}$ the linear space spanned by the vector fields
\begin{equation} \label{eq:clhormander}A_1, \dots, A_m, \quad [A_i,A_j],\,\text{ with } 0\leq i,j \leq m, \quad  [A_i,[A_j,A_k]],\text{ with } 0\leq i,j,k \leq m,\dots
\end{equation}
is the whole $\RR^d$,
\end{center}
where $[\cdot, \cdot]$ stands for the commutator of two vector fields, $A_0(x)=\mathbf{Z}^Tx+\boldsymbol{\zeta}$, and $A_j(x)=\mathbf{B}_{(j)}$ is the constant vector field equal to the $j$-th column of the matrix $\mathbf{B}$ at every point. Therefore, the regularising effect can be seen as a particular instance of Theorem 2.3.3 in \cite{Nu13}.

\bigskip The classical case represents a solid starting point for our analysis, however the quantum case turns out be far harder to tackle: for instance, the simple reasoning behind the equivalence between irreducibility, regularisation properties and the condition in eq. \eqref{eq:classcond} breaks down in the noncommutative framework. Indeed, several criticalities arise: there are not quantum counterparts of delta measures, the quantum counterpart of $S_t$ fails to have several nice properties, and the convex hull of quantum Gaussian states is not dense in the set of normal states. There is also a deeper reason for the proof to fail to extend to the quantum case: surprisingly, it turns out that irreducibility and regularisation properties are no longer equivalent in this setting.

Let us denote by $\TT_t$ the considered GQMS and by $\TT_{t*}$ its predual, acting on normal states; the first main result of the paper is to show that the condition
\begin{equation} \label{eq:realCzcond} \tag{real $\mathbf{C}_Z$-cond}\text{there are no nontrivial $\mathbf{Z}$-invariant subspaces of } \RR^{2d} \text{ contained in } \ker(\mathbf{C}_{Z}),\end{equation}
which we will denote as \textit{real $\textbf{C}_{Z}$-cond}, is equivalent to the fact that for every initial state $\rho$ and for every strictly positive time $t>0$, $\TT_{t*}(\rho_t)$ is \textit{smooth}, meaning that one can make sense of the following operator $\partial_{R_1}\cdots \partial_{R_m}(\TT_{t*}(\rho_t))$ for every finite collection $R_1,\dots, R_m$ of field operators, where $\partial_{x}$ is the commutator with respect to $x$. This is a natural notion of smoothness in the quantum case, as we explain in section \ref{sec:pre}; moreover, we show in Proposition \ref{prop:Sobolev} that one cannot expect better regularisation properties from GQMSs, e.g. the one considered in \cite{RGM25}. Under condition \textit{real $\textbf{C}_{Z}$-cond}, we show that there exists a representation formula for the semigroup that reads
\begin{equation}\label{eq:convo}
\TT_{t*}(\rho)=g_t*{\cal S}_t(\rho),\end{equation}
where $g_t$ is a Gaussian kernel and ${\cal S}_t$ is the quantum counterpart of $S_t$ (see eq.\eqref{eq:clrepresentation}); such a representation formula is particularly useful for deducing quantitative regularity estimates: for instance, we exploit it to deduce some simple estimates on the Hilbert-Schmidt norm of operators of the form $\partial_R(\TT_{t*}(\rho))$, where $R$ is a field operator. Finally, we prove that \textit{real $\textbf{C}_{Z}$-cond} is also equivalent to the fact that the decoherence-free subalgebra ${\cal N}(\TT)$ is trivial, i.e. equal to $\CC\mathbf{1}$; the decoherence-free subalgebra is a relevant object for the study of long-time behaviour of QMSs and for GQMSs is the biggest Von Neumann algebra on which $(\TT_t)_{t \geq 0}$ acts as a group of $^*$-automorphism. We can therefore informally state the first important contribution of the work as follows (see Theorem \ref{theo:regularisation} for the rigorous statement).

\begin{result}
Let $\TT$ be a Gaussian quantum Markov semigroups; then, the following statements are equivalent:
\begin{enumerate}
    \item real $\mathbf{C}_Z$-cond holds true;
    \item $\TT_{t*}(\rho)$ is infinitely differentiable for all positive times $t$ and initial densities $\rho$;
    \item ${\cal N}(\TT)=\CC \mathbf{1}.$
\end{enumerate}
\end{result}

We remark that real $\mathbf{C}_Z$-cond is related to the notion of controllability considered in \cite{Ya14}: indeed, real $\mathbf{C}_Z$-cond is just a restatement of the fact that the controllability matrix is full rank.

Let us now turn to results regarding irreducibility. The infinitesimal generator of GQMSs is in a generalised GKSL form and it can be (formally) written as
$$\LL(X)=i\partial_H(X)-\frac{1}{2}\sum_{\ell=1}^{m}L_\ell^*L_\ell X -2L_\ell^* X L_\ell + X L_\ell^*L_\ell$$
when acting on suitable operators. With some algebraic manipulations one can write ${\cal L}$ as the sum of a drift part ${\cal L}_1$ and a diffusion part
$${\cal L}_2(X)=-\frac{1}{2}\sum_{\ell=1}^m\partial_{L_\ell^*}\partial_{L_\ell}(X),$$
in the sense that $\LL_1$ (resp. $\LL_2$) encodes all the information regarding $\mathbf{Z}$ and $\boldsymbol{\zeta}$ (resp. $\mathbf{C}$). It is natural, therefore, to define a quantum H\"ormander condition (\textit{H\"or-cond}) in terms of $\LL_1$ and $\partial_{L_\ell}$s (for more details see eq. \eqref{eq:Horcond}). Thanks to the convenient splitting of the infinitesimal generator, we are able to prove the equivalence between \textit{$\mathbb{G}$-cond}, \textit{$\textbf{C}_{Z}$-cond} and \textit{H\"or-cond} (see Theorem \ref{th:equivcond}). Then, noticing that \textit{$\textbf{C}_{Z}$-cond} is stronger than \textit{real $\textbf{C}_{Z}$-cond} and leveraging on the regularisation effect discussed above, we are able to show that irreducibility of the semigroup is equivalent to \textit{$\mathbb{G}$-cond}, \textit{$\textbf{C}_{Z}$-cond} and \textit{H\"or-cond} (see Theorem \ref{theo:main}). Indeed, the main technical obstacle faced in previous works (\cite{AFP22b,FP22}) was due to the fact that it was not clear whether one could find ``smooth vectors'' in the range of subharmonic projections, where by ``smooth vectors'' we mean that one can evaluate on them arbitrary polynomials in creation and annihilation operators. We solve this issue using regularisation effects of GQMSs and the properties of smooth operators proved in Proposition \ref{prop:schwarz}. We briefly sum up the second main contribution of our work.

\begin{result}
For a Gaussian quantum Markov semigroup $\TT$, the following statements are equivalent:
\begin{enumerate}
\item \textit{$\mathbb{G}$-cond} holds true;
\item \textit{$\textbf{C}_{Z}$-cond} holds true;
\item \textit{H\"or-cond} holds true;
\item $\TT$ is irreducible.
\end{enumerate}
\end{result}

Conversely to the classical case, irreducibility is a strictly stronger property than having regularising effects: indeed, we can construct explicit examples of GQMSs satisfying \textit{real $\textbf{C}_{Z}$-cond}, but not \textit{$\textbf{C}_{Z}$-cond} (in fact, we can characterise all such GQMSs as discussed after Corollary \ref{coro:irrimplnt}). Finally, we remark that our results show that irreducibility implies that the decoherence-free subalgebra is trivial (Corollary \ref{coro:irrimplnt}), which is a known fact for uniformly continuous semigroups admitting a faithful normal invariant state (see, for instance, Proposition 4.3 in \cite{DFSU16}), but it is not known to hold in full generality.

The structure of the paper is as follows: in section \ref{sec:pre} we set the notation and recall the concepts and results which are needed to understand this work. Section \ref{sec:main} is divided into two subsections: the first one presents the results regarding regularisation effects of Gaussian quantum Markov semigroups and their relationship with the decoherence-free subalgebra, a representation formula using quantum convolution of the semigroup action, and some properties and results concerning smooth operators; the second part provides a description of the literature regarding irreducibility of GQMSs so far and presents the main results concerning the several equivalent characterisations of irreducibility

\section{Preliminaries and notation} \label{sec:pre}
In this section we will set the notation and recall the main definitions and results that are needed in order to read this work.

\textbf{Bosonic Fock space.} Let $\mathfrak{h}:=\FS$ be the symmetric or bosonic Fock space over $\CC^d$; we recall that this is the closed subspace of the free Fock space
$$
\bigoplus_{n\geq 0} (\CC^d)^{\otimes n}
$$
generated by exponential vectors, i.e. those vectors of the form
$$
e_z=\sum_{n \geq 0} \frac{z^{\otimes n}}{\sqrt{n!}}, \quad z \in \CC^d.
$$
In the physical literature normalized exponential vectors $e^{-\|z\|^2/2}e_z$ are usually called coherent vectors and $\ket{e_0}\bra{e_0}$ is known as the vacuum state. We recall that $\FS$ is isometrically isomorphic to $d$ copies of $\Gamma(\CC)$ via the following correspondence:
\begin{align*}
    \FS &\rightarrow \Gamma(\CC) \otimes \cdots \otimes \Gamma(\CC)\\
    e_z &\mapsto e_{z_1} \otimes \cdots \otimes e_{z_d}
\end{align*}
where $z_i$'s are the coordinates of $z \in \CC^d$ in any orthonormal basis. Every $\Gamma(\CC)$ corresponds to a mode, is isomorphic to $\ell^2(\nn)$ and has its own creation, annihilation and number operators; let $\{e(n_1,\dots,n_d):=e(n_1) \otimes \cdots \otimes e(n_d)\}_{n_1,\dots, n_d \in \nn}$ be the canonical orthonormal basis for $\ell^2(\mathbb{N})^{\otimes d}\simeq \FS$, then the annihilation, creation and number operators corresponding to the $j$-th mode $a_j$, $a^\dagger_j$, $N_j$ act in the following way on the basis elements:
\begin{align*}
&a_je(n_1,\dots,n_d)=\sqrt{n_j}\,e(n_1,\dots,n_{j-1}, n_j-1, n_{j+1}, \dots , n_d),\\
&a^\dagger_je(n_1,\dots,n_d)=\sqrt{n_j+1}\,e(n_1,\dots,n_{j-1}, n_j+1, n_{j+1}, \dots,  n_d),\\
&N_je(n_1,\dots,n_d)=n_j \,e(n_1, \dots,  n_d).
\end{align*}
If any $n_j<0$, then $e(n_1, \dots, n_d)$ has to be understood as $0$. We recall that $N_j=a_j^\dagger a_j$ and that the linear space $D$ of finite linear combinations of the vectors of the canonical basis is an essential domain for all such operators.

Another family of operators of fundamental importance in our analysis is the set of Weyl operators: given $z \in \CC^d$, the corresponding Weyl operator is the unique unitary operator acting in the following way on coherent vectors:
\begin{align} \label{eq:Wexpo}
W(z):\FS &\longrightarrow \FS\\
e^{-\frac{\|w\|^2}{2}}e_w &\longmapsto e^{-i\Im(\langle z,w \rangle)} e^{-\frac{\|w+z\|^2}{2}}e_{w+z}.
\end{align}
Weyl operators satisfy the exponential form of canonical commutation relations, i.e.
\begin{equation} \label{eq.expccr}
W(z+w)=e^{i\Im(\langle z,w \rangle)}W(z)W(w), \quad z,w \in \CC^d.
\end{equation}
Moreover, the set of Weyl operators is ${\rm w}^*$-dense in $B(\hh)$. To any $z \in \CC^d$, using Stone's theorem we can associate the unique self adjoint operator $R(z)$ which is the generator of the strongly continuous group $t \mapsto W(tz)$; in this sense, we will write $W(z)=e^{iR(z)}$. $R(z)$'s are called quadratures or field operators; let $\{f_1,\dots,f_d\}$ be the canonical basis of $\CC^d$, then one can check that
$$
R(f_j)=-\sqrt{2}P_j, \quad R(if_j)=\sqrt{2}Q_j, \quad j=1,\dots, d,$$
where $Q_j$ and $P_j$ are position and momentum observables, respectively, corresponding to the $j$-th mode. In general, one can see that
\begin{equation} \label{eq:fields}\sum_{j=1}^{d}\sqrt{2}\Im(z_j) Q_j-\sqrt{2}\Re(z_j) P_j\subseteq R(z), \quad z=(z_1,\dots, z_d) \in \CC^d,
\end{equation}
where the sum is defined on a common domain for the operators involved (e.g. $D$).

\bigskip \textbf{Symplectic structure of $\CC^d$.} Eq. \eqref{eq:fields} hints that, when one deals with field operators, it might be convenient to consider the real Hilbert space structure of $\CC^d$ as well and the following identification with $\RR^{2d}$:
\begin{align*}
\CC^d &\longrightarrow \RR^{2d}\\
x+iy &\longmapsto \begin{pmatrix} x \\ y\end{pmatrix}.
\end{align*}
The real inner product on $\RR^{2d}$ corresponds to
$$\langle z,w \rangle_\RR:=\Re(\langle z,w \rangle), \quad z,w \in \CC^d.$$
We will use the bold font $\mathbf{z}$ to denote the vector in $\RR^{2d}$ corresponding to $z \in \CC^d$. Given a real linear vector subspace $V$ in $\CC^{d}$ we will denote by $\boldsymbol{{\cal V}}$ its image via the identification above, i.e.
$$
\boldsymbol{{\cal V}}:=\left \{ \mathbf{z} \in \RR^{2d} : \, z \in V\right \}.
$$
Let $A$ be a real linear operator on $\CC^d$, one can always write it in the form
\begin{equation} \label{eq:A1A2}
Az=A_1z+A_2\overline{z},
\end{equation}
where $A_1$, $A_2$ are complex linear operators on $\CC^d$. As an operator on $\RR^{2d}$, $A$ reads as
\begin{equation} \label{eq:boldA}
\mathbf{A}=\begin{pmatrix} \Re(A_1)+\Re(A_2) & \Im(A_2)-\Im(A_1) \\
\Im(A_1)+\Im(A_2) & \Re(A_1)-\Re(A_2)
\end{pmatrix}.
\end{equation}
With an abuse of notation we will denote by $\mathbf{A}$ its complexification as well, acting on $\CC^{2d}$ and we will call spectrum of $\mathbf{A}$, denoted by ${\rm Sp}(\mathbf{A})$, the set of those $z \in \CC$ such that $z-\mathbf{A}$ is not invertible as an operator on $\CC^{2d}$.

The imaginary part of the complex inner product on $\CC^d$, which appears in the commutation relations between Weyl operators (Eq. \eqref{eq.expccr}), is a non-degenerate symplectic form on $\CC^d$, i.e. it is a bilinear form satisfying the following two requirements:
\begin{itemize}
\item $\Im(\langle z,w\rangle=-\Im(\langle w,z\rangle)$, $z,w \in \CC^d$;
\item $\Im(\langle z,w\rangle )=0$ for all $w \in \CC^d$ if and only if $z=0$.
\end{itemize}
We will use the notation $\sigma(z,w):=\Im(\langle z,w \rangle)$. One can immediately check that the symplectic form in $\RR^{2d}$ reads as
$$\sigma(z,w)=\langle \mathbf{z}, \mathbf{J} \mathbf{w}\rangle, \text{ where } \mathbf{J}=\begin{pmatrix} \mathbf{0} & \mathbf{1} \\
-\mathbf{1} & \mathbf{0} \end{pmatrix}.$$
$\mathbf{0}$ (resp. $\mathbf{1}$) is the $d\times d$-matrix with all entries equal to $0$ (resp. the $d$-dimensional identity matrix). $\mathbf{J}$ is called symplectic matrix. With an abuse of notation, for every $z \in \CC^d$ we will use also the notation $W(\mathbf{z})$ (resp. $R(\mathbf{z})$) to denote $W(z)$ (resp. $R(z)$).

\bigskip \textbf{Quantum Fourier transform, quantum convolution and smooth operators.} An expository account of the quantum Fourier transform can be found in \cite[Section 5.4]{Ho82}. Let us denote by $L^1(\hh)$ (respectively $L^2(\hh)$) the space of trace class (resp. Hilbert-Schmidt) operators on $\hh$; we denote by $\| \, \|_1$ and $\| \, \|_2$ the corresponding norms; moreover, $L^2(\RR^{2d},2^{d}\mathbf{dz})$ will denote the space of complex valued square integrable functions on $\RR^{2d}$ with respect to $2^{d}$ times the Lebesgue measure. Any trace class operator $T \in L^1(\hh)$ is uniquely determined by the following function, that we call quantum Fourier transform 
$$\widehat{T}(\mathbf{z})=\tr(TW(\mathbf{z})), \quad \mathbf{z} \in \RR^{2d}.$$
If $T$ is a state, $\widehat{T}$ is also known as quantum characteristic function. We recall that the map $T \mapsto \widehat{T}$ extends uniquely to a unitary map from $L^2(\hh)$ to $L^2(\RR^{2d},2^{d}\mathbf{dz}).$

The notions of smoothness and decay of a function in $L^2(\RR^{2d},2^{d}\mathbf{dz})$ along a direction $\mathbf{w}$ can be expressed in terms of the following families of strongly continuous unitary groups:
$$\widehat{{\cal U}}^{\mathbf{w}}_t(f)(\mathbf{z})=e^{i2\langle \mathbf{w},\mathbf{z} \rangle t}f(\mathbf{z}), \quad \widehat{{\cal V}}^{\mathbf{w}}_t(f)(\mathbf{z})=f(\mathbf{z}+ 2\mathbf{w}).
$$
Indeed, the corresponding infinitesimal generators are given by $i$ times the following selfadjoint operators:
$$\widehat{{\cal H}}_{\mathbf{w}}(f)(\mathbf{z})=2\langle \mathbf{w}, \mathbf{z} \rangle f(\mathbf{z}), \quad  \widehat{{\cal K}}_{\mathbf{w}}(f)(\mathbf{z})=2\partial_{\mathbf{w}}f(\mathbf{z}),$$
where
\begin{align} 
&D(\widehat{{\cal H}}_{\mathbf{w}})=\{f \in L^2(\RR^{2d},2^{d}\mathbf{dz}):\,\langle \mathbf{w}, \mathbf{z} \rangle f(\mathbf{z}) \in L^2(\RR^{2d},2^{d}\mathbf{dz})\},\label{eq:domainHhat}\\
&D(\widehat{{\cal K}}_{\mathbf{w}})=\{f \in L^2(\RR^{2d},2^{d}\mathbf{dz}):\,\partial_{\mathbf{w}}f(\mathbf{z}) \in L^2(\RR^{2d},2^{d}\mathbf{dz})\},
\end{align}
and $\partial_{\mathbf{w}}f(\mathbf{z})$ has to be interpreted as the weak derivative of $f$. Therefore, smoothness of a function $f \in L^2(\RR^{2d},2^{d}\mathbf{dz})$ along $\mathbf{w}$ can be characterised as belonging to $D(\widehat{{\cal K}}_{\mathbf{w}})$, while belonging to $D(\widehat{{\cal H}}_{\mathbf{w}})$ corresponds to a certain decay property along the direction $\mathbf{w}$ as .

It is immediate to find the corresponding groups acting on $L^2(\hh)$, i.e. those strongly continuous unitary groups ${\cal U}^{\mathbf{w}}_t$, ${\cal V}^{\mathbf{w}}_t$ that satisfy the following equations: for every $T \in L^2(\hh)$
\begin{equation} \label{eq:equivalenceofgroups}\widehat{{\cal U}}^{\mathbf{w}}_t (\widehat{T})=\widehat{{\cal U}^{\mathbf{w}}_t(T)}, \quad \widehat{{\cal V}}^{\mathbf{w}}_t (\widehat{T})=\widehat{{\cal V}^{\mathbf{w}}_t(T)}.\end{equation}

Indeed, they are given by
$${\cal U}^{\mathbf{w}}_t(T)=W(-t\mathbf{J}\mathbf{w})TW(-t\mathbf{J}\mathbf{w})^*, \quad {\cal V}^{\mathbf{w}}_t(T)=W(t\mathbf{w})TW(t\mathbf{w}),
$$
which are respectively the semigroups of translations in the phase space and symmetric multiplication by the exponential of $i$ times a field operator. The corresponding infinitesimal generators are informally given by $i$ times $[R(-\mathbf{J}\mathbf{w}),\cdot]$ and $\{R(\mathbf{w}),\cdot\}$, respectively. More precisely, we have that
$${\cal H}_{\mathbf{w}}(T)=[R(-\mathbf{J}\mathbf{w}),T], \quad  {\cal K}_{\mathbf{w}}(T)=\{R(\mathbf{w}),T\},$$
where
\begin{align}\label{eq:domains}
&D({\cal H}_{\mathbf{w}})=\{T \in L^2(\hh):\,T(D(R(-\mathbf{J}\mathbf{w}))\subseteq D(R(-\mathbf{J}\mathbf{w})), \overline{[R(-\mathbf{J}\mathbf{w}),T]} \in L^2(\hh)\},\\
&D({\cal K}_{\mathbf{w}})=\{f \in L^2(\RR^{2d}):\,T \in L^2(\hh):\,T(D(R(\mathbf{w}))\subseteq D(R(\mathbf{w})), \overline{\{R(\mathbf{w}),T\}} \in L^2(\hh)\}.
\end{align}
The operators $[R(-\mathbf{Jw}),T]$ and $\{R(\mathbf{w}), T\}$ whose closure is considered in the previous pair of equations has to be intended as defined on $D(R(-\mathbf{Jw}))$ and $D(R(\mathbf{w}))$, respectively. The characterisation of the domain of ${\cal H}_\mathbf{w}$ appeared in \cite{Co82} and a neat exposition can be found in \cite[Theorem 3.6]{LHB24}; the same proof works for ${\cal K}_\mathbf{w}$ as well.

Follow the same principle as in classical harmonic analysis, where a regularity of a function corresponds to decay properties of its Fourier transform and vice versa, we will refer to smoothness and decay of an operator along a direction $\mathbf{w}$, meaning that it belongs to the domain of ${\cal H}_{\mathbf{w}}$ or ${\cal K}_{\mathbf{w}}$, respectively. This is in agreement with the philosophy of quantum harmonic analysis as introduced in the seminal work \cite{We84}. Indeed, let us recall the definition of Schwarz operators introduced in \cite{KKW16}; in order to do so, we need to introduce some more notation. Let ${\cal S}(\RR^d)$ denote the set of Schwarz functions; we recall that Schwarz functions are those infinitely differentiable functions $\varphi:\RR^d \rightarrow \CC$ such that for any pair of multiindices $\boldsymbol{\alpha}, \boldsymbol{\beta} \in \mathbb{N}^{d}$
$$\sup_{\mathbf{x} \in \RR^{d}}|\mathbf{x}^{\boldsymbol{\alpha}} \partial_{\boldsymbol{\beta}}\varphi(\mathbf{x})|<+\infty,$$
where $\mathbf{x}^{\boldsymbol{\alpha}}=x_1^{\alpha_1}\cdots x_d^{\alpha_d}$ and $\partial_{\boldsymbol{\beta}}=\left (\frac{d}{dx_1}\right )^{\beta_1} \cdots  \left (\frac{d}{dx_d}\right )^{\beta_d}$. It is well known that Schwarz functions constitute a dense linear subspace of $L^2(\RR^{d})$ (with respect to the Lebesgue measure) and that belong to the domain of arbitrarily long products of field operators, i.e. for every $\mathbf{z} \in \RR^{2d}$, ${\cal S}(\RR^d) \subset D(R(\mathbf{z}))$ and $R(\mathbf{z})({\cal S}(\RR^d))\subseteq {\cal S}(\RR^d)$. The following alternative characterisation of Schwarz functions will turn out useful later in this manuscript. Given a selfadjoint operator $(X,D(X))$, we introduce the notation $D(X^\infty):=\bigcap_{k \geq 0}D(X^k)$.
\begin{lemma} \label{lem:schwarz}
The following identity of sets holds true:
$$\bigcap_{j=1}^{d} D(P^\infty_j) \cap D(Q^\infty_j)={\cal S}(\RR^d).$$
\end{lemma}
The proof of Lemma \ref{lem:schwarz} can be found in \ref{app:A}. We are now ready to define the space of Schwarz operators as follows:
$${\cal S}(\hh):=\{T \in L^2(\hh): \, \|T\|_{\boldsymbol{\alpha},\boldsymbol{\alpha^\prime},\boldsymbol{\beta},\boldsymbol{\beta^\prime}}<+\infty\, \text{ for all }\boldsymbol{\alpha},\,\boldsymbol{\alpha^\prime},\,\boldsymbol{\beta},\,\boldsymbol{\beta^\prime}\in \mathbb{N}^d\},$$
where $\| \, \|_{\boldsymbol{\alpha},\boldsymbol{\alpha^\prime},\boldsymbol{\beta},\boldsymbol{\beta^\prime}}$ is the seminorm defined as
$$\|T\|_{\boldsymbol{\alpha},\boldsymbol{\alpha^\prime},\boldsymbol{\beta},\boldsymbol{\beta^\prime}}:=\sup\{|\langle P^{\boldsymbol{\beta}}Q^{\boldsymbol{\alpha}}\varphi, TP^{\boldsymbol{\beta^\prime}}Q^{\boldsymbol{\alpha^\prime}}\psi\rangle|,\, \varphi, \psi \in {\cal S}(\RR^d),\|\varphi\|, \, \|\psi\|\leq 1\}.$$

Loosely speaking, multiplication on the left and on the right by arbitrary products of position and momentum operators makes sense on Schwarz operators and by performing such operations one obtains once again a Hilbert-Schmidt operator. Schwarz operators can equivalently be characterised using the following seminorms involving the generators ${\cal H}_\mathbf{w}$ and ${\cal K}_{\mathbf{w}}$:
$$\|T\|^\prime_{\boldsymbol{\alpha},\boldsymbol{\beta}}:=\sup\{|\tr({\cal H}_{\boldsymbol{\alpha}}{\cal K}_{\boldsymbol{\beta}}(\ket{\psi}\bra{\varphi})T)|: \, \|\psi\|, \|\varphi\| \leq 1\},$$
where $\boldsymbol{\alpha}, \boldsymbol{\beta} \in \mathbb{N}^{2d}$, ${\cal H}_{\boldsymbol{\alpha}}={\cal H}^{\alpha_1}_{e_1}\cdots {\cal H}_{e_{2d}}^{\alpha_{2d}}$ and ${\cal K}_{\boldsymbol{\beta}}$ is defined analogously.

There is a last ingredient that we will need in the following: the convolution of a function with respect to an operator. Given an operator $T \in L^2(\hh)$, we explained why $W(-\mathbf{J}\mathbf{x})TW(\mathbf{J}\mathbf{x})$ can be seen as the translation of $T$ by $\mathbf{x}$; therefore, the following definition of convolution should look quite natural: given $f \in L^1(\RR^{2d})$ we define
$$f*T=\int_{\RR^{2d}}f(\mathbf{x})W(-\mathbf{J}\mathbf{x})TW(\mathbf{J}\mathbf{x}).$$
As shown in \cite{We84}, $f*T$ is an Hilbert-Schmidt (resp. trace class) operator if $T$ is an Hilbert-Schmidt (resp. trace class) operator and it turns out that such a definition of convolution shares several useful properties with the classical counterpart.

\bigskip \textbf{Gaussian quantum Markov semigroups.} A Quantum Markov Semigroup (QMS) $\TT:=\{\TT_t\}_{t \geq 0}$ is a ${\rm w}^*$-continuous semigroup of completely positive, identity preserving, ${\rm w}^*$-continuous maps on $B(\hh)$. The predual semigroup $\TT_* = \{\TT_{t*}\}_{t\geq0}$ acts on trace class operators and is a strongly continuous, completely positive, trace preserving semigroup. A quantum Markov semigroup is called Gaussian if it maps Gaussian states into Gaussian states; the class of Gaussian quantum Markov semigroups (GQMSs) can be completely characterised either through their explicit action on Weyl operators or through their generator. Let $L_\ell$, $H$ be the operators on $\hh$ defined on the domain $D$ by the following expressions:
\begin{align}
&H=\sum_{k,j=1}^{d} \left ( \Omega_{jk} a^\dagger_j a_k + \frac{\kappa_{jk}}{2}a_j^\dagger a_k^\dagger+\frac{\overline{\kappa}_{jk}}{2}a_j a_k\right )+\frac{1}{2}\sum_{j=1}^{d} \zeta_j a_j^\dagger + \overline{\zeta}_j a_j,\\ \label{eq:hami}
&L_\ell=\sum_{j=1}^{d} \overline{v}_{\ell j}a_j+u_{\ell j}a_j^\dagger, \quad \ell=1,\dots, m
\end{align}
where $\Omega \in M_d(\CC)$ is hermitian, $\kappa \in M_d(\CC)$ is symmetric, $\zeta \in \CC^d$, $m \leq 2d$ and $U,V \in M_{m\times d}(\CC)$. One can show that $H$ and $L_\ell$ are closable and we will denote their closure with the same name, with a slight abuse of notation. For all $x \in B(\hh)$ consider the following quadratic form with domain $D\times D$
\begin{equation} \label{eq:fg} \begin{split}\mathfrak{L}(x)[\xi^\prime, \xi] &= i \langle H\xi^\prime, x\xi\rangle -i \langle \xi^\prime, xH\xi \rangle\\
&-\frac{1}{2}\sum_{\ell=1}^{m}(\langle \xi^\prime,xL^*_\ell L_\ell\xi \rangle - 2 \langle L_\ell\xi^\prime,x L_\ell\xi \rangle +\langle L^*_\ell L_\ell\xi^\prime,x\xi \rangle). \end{split}\end{equation}

This is a natural way to make sense of a Gorini, Kossakowski, Lindblad-Sudarshan
(GKLS) representation of the generator in a generalized form since operators $L_\ell$, $H$ are unbounded. The following result ensures that the form generator in Eq. \eqref{eq:fg} generates a quantum Markov semigroup and provides its action on Weyl operators; its proof can
be found in \cite{AFP22},
Theorem 2 in Appendix A and Theorem 2.4

\begin{theo}
There exists a unique quantum Markov semigroup, $\TT$ such that, for all $x \in B(\hh)$ and $\xi, \xi^\prime \in D$, the function $t \mapsto \langle \xi^\prime, \TT_t(x) \xi \rangle$ is differentiable and
$$
\frac{d}{dt} \langle \xi^\prime, \TT_t(x) \xi \rangle = \mathfrak{L}(\TT_t(x))[\xi^\prime,\xi], \quad  \forall t \geq 0.
$$
Moreover,
\begin{equation} \label{eq:aW}
\TT_t(W(\mathbf{z}))=\exp\left (-\frac{1}{2}\int_0^t \langle e^{sZ}z,Ce^{sZ}z \rangle_\RR ds + i \int_0^t \langle \zeta,e^{sZ}z \rangle_\RR ds\right )W(e^{tZ}z),
\end{equation}
where 
\begin{align}
&Zz=[(U^T\overline{U}-V^T\overline{V})/2+i \Omega]z+ [(U^TV-V^TU)/2+i \kappa]\overline{z},\label{eq:Z}\\
&Cz=(U^T\overline{U}+V^T\overline{V})z + (U^TV+V^TU)\overline{z} \label{eq:C}.
\end{align}
\end{theo}
One can check that $Z$ and $C$ satisfy the following inequality, which ensures that Heisenberg uncertainty relations are preserved:
$$\mathbf{C}_{Z}:=\mathbf{C}-i(\mathbf{Z}^*\mathbf{J}+\mathbf{J}\mathbf{Z}) \geq 0,$$
where the positivity condition is the one for complex linear operators on $\CC^{2d}$. For the sake of a more compact presentation, let us introduce the following notation:
$$\mathbf{C}_t:=\int_0^{t} e^{s\mathbf{Z}^*}\mathbf{C}e^{s\mathbf{Z}} ds, \quad \boldsymbol{\zeta}_t:=\int_0^te^{s\mathbf{Z}^*}\boldsymbol{\zeta}ds.$$
The action on the semigroup on Weyl operators becomes
$$\TT_t(W(\mathbf{z}))=e^{-\frac{\langle \mathbf{z},\mathbf{C}_t\mathbf{z}\rangle}{2}+i\langle\boldsymbol{\zeta}_t, \mathbf{z}\rangle}W(e^{t\mathbf{Z}}\mathbf{z}), \quad t \geq 0.$$

A central object in the study of asymptotic properties of $\TT$ is the decoherence-free subalgebra $\NN$, which is defined in the following way:
$$\NN:=\{x \in B(\hh):\TT_t(x^*x)=\TT_t(x^*)\TT_t(x), \, \TT_t(xx^*)=\TT_t(x)\TT_t(x^*), \, t \geq 0\}.$$
The decoherence-free subalgebra is the biggest Von Neumann subalgebra on which $\TT$ acts as a group of $*$-automorphisms (see \cite[Proposition 3]{GP26}) and it can be described in a neat way in terms of Weyl operators.

\begin{prop}[Corollary 14, \cite{AFP22}] \label{prop:decofree}
Let $\boldsymbol{{\cal V}}\subseteq \RR^{2d}$ be the biggest $\mathbf{Z}$-invariant subspaces in $\ker(\mathbf{C})$, then
$$\NN=\{W(\mathbf{z}):\, \mathbf{z} \in \boldsymbol{{\cal V}}\}^{\prime \prime}.$$  
\end{prop}

\bigskip \textbf{Irreducibility.} We say that a QMS $\TT$ is irreducibile if there exists no nontrivial orthogonal projections $p$ such that
$\TT_t(p) \geq p$ for every $t \geq 0$; such a $p$ is called subharmonic projection. This is the natural quantum extension of the notion of irreducibility for classical Markov semigroups and, as in the classical case, the name irreducibility is motivated by the fact that one can show (see Proposition 2.5 in \cite{FG25}) that any subharmonic projection $p$ reduces the semigroup in the following sense: if one considers a state $\rho$ such that ${\rm supp}(\rho)\subseteq {\rm supp}(p)$, then for every $t \geq 0$ one has
$${\rm supp}(\TT_{t*}(\rho))\subseteq {\rm supp}(p).$$

We recall that for every positive semidefinite operator $x$, ${\rm supp}(x)$ is defined as the orthogonal complement of $\ker(x)$. We refer the interested reader to the recent review that we wrote (\cite{FG25}) for further details. If a QMS $\TT$ admits a generalised GKLS-form, as Gaussian semigroups do, there is a characterisation of subharmonic projections in terms of the operators appearing in the generator (see Theorem III.1 in \cite{FR02}); we will present the result in the specific case of GQMSs in order to avoid the introduction of some technical hypotheses which are always satisfied by GQMSs. Let us introduce the operator defined as
\begin{equation} \label{eq:G}
G:=-iH-\frac{1}{2}\sum_{\ell=1}^{m}L_\ell^*L_\ell\end{equation}
on $D$; one can show that it is closable and its closure, denoted by the same symbol with a slight abuse of notation, generates a strongly continuous contraction semigroup $P=(P_t)_{t \geq 0}$ on $\hh$. Moreover, $D(G) \subseteq D(L_\ell)$ for every $\ell=1,\dots, m$. (see Appendix A in \cite{AFP22}). 

\begin{theo}[Theorem III.1, \cite{FR02}] \label{theo:subharmonic}
A projection $p$ is subharmonic for a GQMS $\TT$ if and only if ${\rm supp}(p)$ is invariant for the operators $P_t$ for every $t \geq 0$ and
$$L_\ell u =pL_\ell u, \quad \ell=1,\dots, m,$$
for every $u \in D(G)\cap {\rm supp}(p).$
\end{theo}

We recall that semigroup theory ensures that $D(G)\cap {\rm supp}(p)$ is dense in ${\rm supp}(p)$.

\section{Main results} \label{sec:main}

Along this section $\TT$ will always denote a Gaussian quantum Markov semigroup.

\subsection{Regularisation properties and decoherence-free subalgebra.}

In the previous section we defined what we mean in this work by smooth Hilbert-Schmidt operator and we showed that it is equivalent to its quantum Fourier transform to have some decay properties. If one starts with a given initial state $\rho$, the density matrix at time $t$, i.e. $\TT_{t*}(\rho)$, has a quantum characteristic function that reads as follows:
\begin{equation} \label{eq:evolution}\widehat{\rho}_t(\mathbf{z})=e^{-\frac{\langle \mathbf{z},\mathbf{C}_t\mathbf{z}\rangle}{2}+i\langle\boldsymbol{\zeta}_t, \mathbf{z}\rangle}\widehat{\rho}(e^{t\mathbf{Z}}\mathbf{z}).\end{equation}
The key observation when studying regularisation properties of GQMSs is that the action of the semigroup introduces an exponential factor in front of the characteristic function of the initial density operator which causes a smoothing effect along the directions which are in ${\rm supp}(\mathbf{C}_t)$. The following Lemma shows that directions in which there is a smoothing effect do not depend on time and connects them to the decoherence-free subalgebra and \ref{eq:realCzcond}.

\begin{lemma} \label{lem:kerct}
For every $t >0$, the following linear subspaces of $\RR^{2d}$ coincide:
\begin{enumerate}
    \item $\ker(\mathbf{C}_t)$,
    \item the biggest $\mathbf{Z}$-invariant linear subspace in $\ker(\mathbf{C})$,
    \item the biggest $\mathbf{Z}$-invariant linear subspace in $\ker(\mathbf{C}_{Z})$.
\end{enumerate}
\end{lemma}
\begin{proof}
First of all, let us remark that
$$\int_0^t e^{s\mathbf{Z}^*}(\mathbf{Z}^*\mathbf{J} + \mathbf{J}\mathbf{Z})e^{s \mathbf{Z}}ds=e^{t\mathbf{Z}} \mathbf{J} e^{t\mathbf{Z}}-\mathbf{J};$$
therefore for every $\mathbf{z}
\in \RR^{2d}$ one has
$$\langle \mathbf{z}, \int_0^t e^{s\mathbf{Z}^*}(\mathbf{Z}^*\mathbf{J} + \mathbf{J}\mathbf{Z})e^{s \mathbf{Z}}ds \,\mathbf{z} \rangle=\langle e^{t \mathbf{Z}}\mathbf{z}, \mathbf{J} e^{t\mathbf{Z}}\mathbf{z} \rangle -\langle \mathbf{z}, \mathbf{J} \mathbf{z} \rangle=0.$$
Hence, one has that $\ker(\mathbf{C}_t)$ coincide with
$$\ker \left ( \int_0^t e^{s\mathbf{Z}^*}\mathbf{C}_{Z}e^{s \mathbf{Z}}ds\right )\cap \RR^{2d}.$$

In order to conclude it is enough to show that for every positive semidefinite matrix $\mathbf{A}$ and for every matrix $\mathbf{B}$,
$$\ker\left ( \int_0^t e^{s\mathbf{B}^*} \mathbf{A} e^{s \mathbf{B}}ds \right )$$
coincides with the biggest $\mathbf{B}$-invariant subspace in $\ker(\mathbf{A})$; indeed, the statement follows considering $\mathbf{A}=\mathbf{C}$ ($\mathbf{A}=\mathbf{C}_Z$) and $\mathbf{B}=\mathbf{Z}$. This is well known, but we report here the proof for completeness. Let us fix $\mathbf{z} \in \RR^{2d}$, then
\begin{align*}&\langle \mathbf{z},\int_0^t e^{s\mathbf{B}^*} \mathbf{A} e^{s \mathbf{B}}ds \mathbf{z} \rangle=0 \Leftrightarrow  \langle e^{s \mathbf{B}} \mathbf{z},  \mathbf{A} e^{s \mathbf{B}}\mathbf{z} \rangle=0, \, 0\leq s \leq t\\
&\Leftrightarrow e^{s\mathbf{B}} \mathbf{z} \in \ker(\mathbf{A}), \, 0 \leq s \leq t \Leftrightarrow \mathbf{B}^k \mathbf{z} \in \ker(\mathbf{A}), \, k \geq 0.\end{align*}
The second implication is due to the continuity and positivity of $s \mapsto \langle e^{s \mathbf{B}} \mathbf{z},  \mathbf{A} e^{s \mathbf{B}}\mathbf{z} \rangle$, while the last one can be obtained in one sense differentiating and in the other using the Taylor series of the matrix exponential function.
\end{proof}

 We recall that we denoted by $\boldsymbol{{\cal V}}\subseteq \RR^{2d}$ the subspace corresponding to the equivalent characterisations provided in the previous Lemma. We are ready to state the main theorem concerning regularisation properties of GQMSs.

\begin{theo}[Regularisation property of GQMSs] \label{theo:regularisation}
Let us consider any initial density $\rho$, then for every $t>0$ the following holds true:
\begin{equation} \label{eq:regularisation}
\TT_{t*}(\rho) \in \bigcap_{\mathbf{w} \in \boldsymbol{{\cal V}}^\perp}D({\cal H}_\mathbf{w}^\infty).
\end{equation}

Moreover, the following statements are equivalent:
\begin{enumerate}
    \item $\boldsymbol{{\cal V}}=\{0\}$;
     \item the decoherence-free subalgebra is trivial, i.e. $\NN=\CC \mathbf{1}$;
     \item for every initial density $\rho$ one has that
     $$\TT_{t*}(\rho) \in D({\cal H}^\infty_{\mathbf{w}}), \quad  \text{for every }t>0, \mathbf{w} \in \RR^{2d};$$
     \item $\mathbf{C}_t>0$ and for every initial density $\rho$ and positive time $t>0$, one has $\TT_{t*}(\rho)=g_t*{\cal S}_t(\rho)$, where
        $$g_t(\mathbf{x})=\frac{1}{(2\pi \det(\mathbf{C}_t))^d}e^{-\frac{\langle \mathbf{x}-\boldsymbol{\zeta}_t,\mathbf{C}_t^{-1}(\mathbf{x}-\boldsymbol{\zeta}_t)\rangle}{2}} \text{ and }\widehat{{\cal S}_t(\rho)}(\mathbf{z}):=\widehat{\rho}(e^{t\mathbf{Z}}\mathbf{z}).$$
\end{enumerate}
\end{theo}

The proof of Theorem \ref{theo:regularisation} can be found in \ref{app:schwarzII}. An interesting line of research would be trying to understand alternative characterisations and the consequences of being a smooth operator; below we limit ourself to present few simple properties and to compare the notion of regularity that we considered with other ones appeared in the literature. The following result provides a useful property of smooth operators, namely that they preserve the set of Schwarz functions; moreover, leveraging on the useful representation of the semigroup in terms of the convolution with a Gaussian density given in item $4.$ of the previous Theorem, we show how to obtain quantitative estimates on the Hilbert-Schmidt norm of directional derivatives of states regularised by the semigroup. We believe that such a representation formula might have a broader use.

\begin{prop} \label{prop:schwarz}
    If ${\cal V}=\{0\}$, then for every $t>0$ and for every initial density $\rho$ one has
    \begin{enumerate}
        \item $\TT_{t*}(\rho) ({\cal S}(\RR^{d})) \subseteq {\cal S}(\RR^{d});$
        \item for every $\mathbf{w} \in \RR^{2d}$
        $$\|{\cal H}_{\mathbf{w}}(\TT_{t*}(    \rho))\|_2\leq e^{-t\tr(\mathbf{Z})}\|\mathbf{w}\|\|\mathbf{C}_t^{-1/2}\|\mathbb{E}[\|\mathbf{X}\|]\|\rho\|_2,$$
        where $\mathbf{X}$ is a standard multivariate Gaussian random vector in $\RR^{2d}$, $\|X\|$ and $\|\mathbf{w}\|$ denote the usual Euclidean norms, and $\|\mathbf{C}_t^{-1/2}\|$ is the induced operator norm.
    \end{enumerate}
\end{prop}

The proof of Proposition \ref{prop:schwarz} can be found in \ref{app:schwarzII}. We remark that the behaviour of ${\bf C}_t$ when $t$ approaches $0$ is completely understood (see for instance Theorem III.2 in \cite{SP23}). We remark that the representation formula in item $2.$ is the quantum counterpart of eq. \eqref{eq:convo}; while the convolution is still a well behaved operation in the noncommutative case, the same does not hold for ${\cal S}_t$, since it fails to be a $^*$-automorphism on $L^1(\hh)$ (using quantum Bochner theorem - Proposition 3.4 in \cite{We84} - one can easily come up with a state which is mapped to an Hilbert-Schmidt operator which is not a state). This issue causes several elementary proofs for classical Ornstein-Uhlenbeck semigroups to fail to extend to the noncommutative setting.

We recall that another notion of regularity for states on the Bosonic Fock space appeared in \cite{RGM25}; for simplicity, let us consider just one mode. The authors define the $k$-th Bosonic Sobolev space for $k \in \nn$ as the set
$$W^{k}:=\{x \in L^1(\hh): x=(\mathbf{1}+N)^{-k/2}y(\mathbf{1}+N)^{-k/2} \text{ for some } y \in L^1(\hh)\}$$
endowed with the norm
$$\|x\|_{W^{k}}:=\|(\mathbf{1}+N)^{k/2}x(\mathbf{1}+N)^{k/2}\|_1,$$
which turns it into a Banach space; by $N$ we denote the number operator. They call a QMS $\TT$ $k$-Sobolev regularising if for every $t>0$
$$\sup_{\|x\|_1\leq 1}\|\TT_t(x)\|_{W^{k}}<+\infty$$
and they say that $\TT$ is Sobolev regularising if it is $k$-Sobolev regularising for every $k \geq 1$; for instance the QMS with generator having only one jump operator of the form $L=a^k$ for any $k>1$ is Sobolev regularising (see Lemma 4.1 in \cite{RGM25}). However, as the following lemma shows, one cannot expect any GQMS to be Sobolev regularising.

\begin{prop} \label{prop:Sobolev}
    GQMSs cannot be $1$-Sobolev regularising.
\end{prop}

The proof in one mode can be found in  \ref{app:Sobo}; the extension to many modes is immediate. On the other hand, GQMSs can be Sobolev preserving (see \cite{GMR24}). We conclude this section pointing out a connection between the condition $\boldsymbol{{\cal V}}=\{0\}$ and regularity theory for second order differential operators. One can see from Eq. \eqref{eq:evolution} that the characteristic function of suitable initial states $\rho$ satisfies the following partial differential equation:
\begin{equation} \label{eq:cfevolution}
\partial_t\widehat{\rho}_t(\mathbf{z})=\left ( -\frac{1}{2} \langle \mathbf{z},\mathbf{C}\mathbf{z}\rangle +i \langle \boldsymbol{\zeta}, \mathbf{z} \rangle \right )\widehat{\rho}_t(\mathbf{z})+\langle \nabla\widehat{\rho}_t(\mathbf{z}),\mathbf{Z}\mathbf{z} \rangle.
\end{equation}
Let us consider the Wigner function corresponding to $\rho_t$, i.e. the inverse (classical) Fourier transform of $\widehat{\rho}(\mathbf{z})$; for those states such that $\hat{\rho}(\mathbf{z})$ is integrable, the Wigner function reads a follows:
$$\rho^{\rm w}_t(\xi,\eta)=\frac{1}{\pi^{2d}}\int_{\RR^{2d}}e^{\sqrt{2}i(\eta x -\xi y)}\widehat{\rho}(x+iy)dxdy.$$
In terms of the Wigner function, eq. \eqref{eq:cfevolution} translates into
\begin{equation}\label{eq:wevo}
    \partial\rho^{\rm w}_t(\xi,\eta)=\frac{1}{2}\tr(H\rho^{\rm w}_t(\xi,\eta)\mathbf{C}^{\rm w}) -{\rm div}( \rho^{\rm w}_t(\xi,\eta)\cdot (\mathbf{Z}^{\rm w T}(\xi,\eta)+\boldsymbol{\zeta}^{\rm w})),
\end{equation}
where
$$\mathbf{C}^{\rm w}=\frac{1}{2}\begin{pmatrix} -{\bf1 }& \mathbf{0}\\
{\bf 0} & {\bf 1}\end{pmatrix}\mathbf{C}\begin{pmatrix} -{\bf1 }& \mathbf{0}\\
{\bf 0} & {\bf 1}\end{pmatrix}, \quad \mathbf{Z}^{\rm w}=\begin{pmatrix} -{\bf1 }& \mathbf{0}\\
{\bf 0} & {\bf 1}\end{pmatrix}\mathbf{Z}\begin{pmatrix} -{\bf1 }& \mathbf{0}\\
{\bf 0} & {\bf 1}\end{pmatrix}, \quad \boldsymbol{\zeta}^{\rm w}=\begin{pmatrix} -{\bf1 }& \mathbf{0}\\
{\bf 0} & {\bf 1}\end{pmatrix}\boldsymbol{\zeta}/\sqrt{2}.$$
Eq. \eqref{eq:wevo} is sometimes referred to as the quantum Fokker-Plank equation (see for instance \cite{AF08}) and, for the particular models we are studying, it coincides with the Fokker-Plank equation of a classical Ornstein-Uhlenbeck process (eq. \eqref{eq:OUfp}). One can immediately see that the condition $\boldsymbol{\cal V}=\{0\}$ is equivalent to the (classical) H\"ormander condition for eq. \eqref{eq:wevo}. We remark that the equivalence between H\"ormander condition for eq. \eqref{eq:wevo} and the fact that $\mathbf{C}_t>0$ for all $t \geq 0$ was already proved in Theorem III.2 \cite{SP23} following a different path.

\bigskip \textit{Irreducibility.} Partial results concerning irreducibility of GQMSs were obtained in \cite{AFP21,FP22,AFP22b,GP26}. In this section we will recall them and provide a unifying picture, proving different equivalent characterisations of irreducibility in terms of $\mathbf{Z}$, $\mathbf{C}$ and $\boldsymbol{\zeta}$, and $H$ and $L_\ell$'s. 

Let us present the first condition whose relationship with respect to irreducibility is investigated in \cite{AFP21,FP22,AFP22b}. Notice that one can identify jump operators $L_\ell$ with vectors in $\CC^{2d}$ in the following way:
$$L_\ell=\sum_{j=1}^{d}\overline{v}_{\ell j}a_j+u_{\ell j}a_j^\dagger \Leftrightarrow \ket{l_\ell}:=\begin{pmatrix}\overline{V}^{(\ell)T}\\
U^{(\ell)T}\end{pmatrix}.$$
One can consider formally or in the weak sense the commutator $\partial_G$ (the operator $G$ was defined in eq. \eqref{eq:G}) applied to any linear combination of $a_j$ and $a^\dagger_j$ for $j=1,\dots, d$, which turns out to be a linear expression in creation, annihilation and identity operators. Therefore, ignoring multiple of the identity operator, we can associate to $\partial_G$  a unique linear operator on $\CC^{2d}$ that we denote $\mathbb{G}$.

The first condition is the following one:

\begin{equation} \tag{$\mathbb{G}$-cond} \label{eq:Gcond}
{\rm span}_{\CC}\{\mathbb{G}^{k}\ket{l_\ell}: \, \ell=1, \dots, m, \, k\geq 0\}=\CC^{2d}.
\end{equation}

We denote by $N$ the total number operator, i.e. $N=\sum_{j=1}N_j$ where $N_j$ have been defined in the introduction. In the mentioned works the following implications were shown:
\begin{enumerate}
\item if $d=1$, then the irreducibility of $\TT$ is equivalent to ($\mathbb{G}$-cond);
\item if $d>1$ and $D(G)=D(N)$, then ($\mathbb{G}$-cond) implies the irreducibility of $\TT$;
\item if $d>1$ and ${\rm span}_{\CC}\{\ket{l_\ell}: \, \ell=1, \dots, m, \}=\CC^{2d}$, then $\TT$ is irreducible;
\item if $d>1$, $D(G)=D(N)$ and another technical condition holds (Hypothesis FQN in \cite{FP22}), then irreducibility of $\TT$ implies \eqref{eq:Gcond}.
\end{enumerate}
We remark that condition ($\mathbb{G}$-cond) is inspired by the fact that in finite dimensional systems such requirement is equivalent to irreducibility (see for instance Theorem 4.2 and Lemma C.2 in \cite{FG25}). Moreover, it was pointed out by the authors that \eqref{eq:Gcond} is reminescent of H\"ormander's condition, since they both involve iterated commutants; in the following, we will clarify such a relationship.

From a completely different perspective, \cite{GP26} (see Proposition 15) showed that, if $\TT$ admits a normal faithful invariant state, then irreducibility is equivalent to following alternative condition:
\begin{equation}  \tag{$\mathbf{C}_{Z}$-cond} \label{eq:CZcond}
\text{there are no no nontrivial $\mathbf{Z}$-invariant subspaces of $\CC^{2d}$ contained in $\ker(\mathbf{C}_{Z})$.}
\end{equation}

Finally, let us introduce the last condition, which can be considered as a noncommutative counterpart of H\"ormander condition for hypoellipticity; all the expressions that we will write in the following make sense when evaluated on suitable operators, for instance on Schwarz ones. First of all, we write the generator of $\TT$ as the sum of a drift and a diffusion parts; given a possibly unbounded operator $X$, we introduce the following notation when it makes sense:
$$\partial_X(Y):=[X,Y], \quad  R_X(Y):=YX, \quad L_X(Y):=XY$$
and $\partial^*_X:=\partial_{X^*}$, $R^*_X=R_{X^*}$, $L^*_X=L_{X^*}$; the notation is motivated by the fact that the last set of operators corresponds formally to the Hilbert-Schmidt adjoints of the previously introduced operators.

Schwarz operators belong to the domain of the ${\rm w}^*$-infinitesimal generator of $\TT$, which we will denote by $\LL$, and
$$\LL(X)=i\partial_H(X)-\frac{1}{2}\sum_{\ell=1}^{m}L_\ell^*L_\ell X -2L_\ell^* X L_\ell + X L_\ell^*L_\ell, \qquad X \in {\cal S}(\hh).$$
Notice that we can rewrite
$$-L_\ell^*L_\ell X +2L_\ell^* X L_\ell -X L_\ell^*L_\ell=-\partial_{L_\ell}^*\partial_{L_\ell}(X)+R_{L_\ell}\partial_{L_\ell}^*(X)-R^*_{L_\ell}\partial_{L_\ell}(X),$$
therefore one has that $\LL$ is equal to the sum of the drift term
$$\LL_1(X)=i\partial_H(X)+\frac{1}{2}\sum_{\ell=1}^{m}(R_{L_\ell}\partial_{L_\ell}^*(X)-R_{L_\ell}^*\partial_{L_\ell}(X))$$
and the diffusion term
$$\LL_2(X)=-\frac{1}{2}\sum_{\ell=1}^m\partial^*_{L_\ell}\partial_{L_\ell}(X).$$
As in the classical case, if we apply (formally or in a weak sense) $\LL_1$ to any field operator, we obtain again a field operator; moreover, one can check (see \ref{app:l1l2}) that the linear operator corresponding to the action of $\LL_1$ on field operators is equivalent to applying the drift matrix $\mathbf{Z}$ and translating by $\boldsymbol{\zeta}$. On the other hand, applying $\LL_2$ to field operators, we obtain $0$. In the noncommutative case, commutators can be seen as the counterpart of vector fields, therefore $\LL_2$ can be seen as the sum of squares of vector fields. We are ready to introduce the noncommutative version of H\"ormander condition:
\begin{equation}\tag{H\"or-cond} \label{eq:Horcond}
    {\rm span}_{\CC}\{\partial_{L_\ell}, [\LL_1,\partial_{L_\ell}], [\LL_1,[\LL_1,\partial_{L_\ell}]], \dots\}_{\ell=1}^{m}={\rm span}_{\CC}\{\partial_{a_j}, \partial^*_{a_j}\}.
\end{equation}

It turns out that all three conditions we presented are equivalent.

\begin{theo} \label{th:equivcond}
Conditions \eqref{eq:Gcond}, \eqref{eq:CZcond} and \eqref{eq:Horcond} are equivalent.
\end{theo}

The proof can be found in \ref{app:C}. Moreover, the equivalent conditions  in Theorem \ref{th:equivcond} characterise  also irreducibility of $\TT.$ 

\begin{theo} \label{theo:main}
The following statements are equivalent:
\begin{enumerate}
\item there are no nontrivial $\mathbf{Z}$-invariant subspaces of $\CC^{2d}$ contained in $\ker(\mathbf{C}_{Z})$;
\item $\TT$ is irreducibile.
\end{enumerate}
\end{theo}
We defer the proof to \ref{app:maintheo}. As an immediate consequence of Theorems \ref{theo:regularisation} and \ref{theo:main}, one obtains the following, which in general is known to hold for any uniformly continuous QMS with a faithful normal invariant state (Theorems 4.2 and 5.8 in \cite{FG25} and Theorem 1 in \cite{CJ20}).

\begin{coro} \label{coro:irrimplnt}
    If $\TT$ is irreducible, then $\NN=\CC\mathbf{1}$.
\end{coro}

The proof of Theorem \ref{theo:main} together with the analysis conducted in \cite{AFP21} shows that irreducibility is a stronger condition than $\NN=\CC\mathbf{1}$ and that every semigroup with trivial decoherence-free subalgebra which is not irreducibile admits an invariant subalgebra which is isomorphic to the algebra of bounded linear operators on a one-mode Fock space, i.e. $\Gamma(\CC)$ on which it has an action which is unitarily equivalent to one of the following two alternatives: the one-mode Hamiltonian does not have any squeezing term, i.e. is of the form
$$H=\Omega a^\dagger a +\frac{\zeta}{2}a^\dagger+\frac{\overline{\zeta}}{2}a,$$
and there is a single jump operator which is a multiple either of the creation operator, or the annihilation operator, i.e.
$$L=u a^\dagger, \text{ or }L=\overline{v} a.$$
Indeed, in the case of a creation-type jump operator one has
$$\mathbf{Z}=\begin{pmatrix} |u|^2/2 & -\Omega\\
\Omega & |u|^2/2\end{pmatrix}, \quad \mathbf{C}=|u|^2\mathbf{1}, \quad   \mathbf{C}_{Z,-} =\begin{pmatrix} 1 & -i \\
i & 1\end{pmatrix}$$
and one immediately sees that $\mathbf{C}>0$, while $\ker(\mathbf{C}_{Z,-})=\CC (i,1)^T$, which is $\mathbf{Z}$-invariant and corresponds to complex multiples of the creation operator. On the other hand, in the case of annihilation-type jump operator one has
$$\mathbf{Z}=\begin{pmatrix} -|v|^2/2 & -\Omega\\
\Omega & -|u|^2/2\end{pmatrix}, \quad \mathbf{C}=|v|^2\mathbf{1}, \quad   \mathbf{C}_{Z,-} =\begin{pmatrix} 1 & i \\
-i & 1\end{pmatrix}$$
and $\mathbf{C}>0$, while $\ker(\mathbf{C}_{Z,-})=\CC (-i,1)^T$, which is $\mathbf{Z}$-invariant and corresponds to complex multiples of the annihilation operator. 

\section{Conclusion}

In this work, we first identified the appropriate notion of regularity for operators in this setting, providing a natural way to formulate and analyse the smoothing effects produced by the GQMSs. Within this framework, we characterised the regularisation effect through algebraic constraints involving the drift and quantum diffusion matrices. These conditions establish a direct connection with the controllability theory of quantum linear systems, while also revealing a close relationship with the structure of decoherence-free subsystems.

We further investigated irreducibility and characterised it through several equivalent algebraic conditions : one expressed in terms of the drift and quantum diffusion matrices, another one in terms of the operators appearing in the generalised GKLS form of the generator and a third one which consists in a quantum counterpart of H\"ormander's condition.

A surprising central consequence of this analysis is that irreducibility is strictly stronger than regularisation properties, conversely to the classical case. While the condition characterising regularisation properties coincides with controllability for the quantum linear system and with hypoellipticity of the equation satisfied by the Wigner function of the initial state, it would be interesting to find an interpretation in control theory and in the theory of partial differential equations for the stronger irreducibility assumption.

The results presented here should therefore be viewed as a starting point for a broader analysis of quantum Markov semigroups acting on continuous variable systems. In particular, they provide a basis for the systematic study of reducible GQMSs and suitable perturbations of irreducible GQMSs. The second task inevitably demands a deeper study of the notion of smooth operator that we considered in this work.

\section*{Acknowledgement}

The authors have been partially supported by the MUR grant Dipartimento di Eccellenza
2023–2027 of Dipartimento di Matematica, Politecnico di Milano. F.G. is a member of the INdAM-GNAMPA group.

\appendix

\section{Proof of Lemma \ref{lem:schwarz}} \label{app:A}

\begin{proof}[Proof of Lemma \ref{lem:schwarz}]
The inclusion ${\cal S}(\RR^d) \subseteq \bigcap_{j=1}^{d} D(P^\infty_j) \cap D(Q^\infty_j)$ follows by the characterisation of Schwarz functions given, for instance, in Lemma 2.1 in \cite{KKW16}; let us now prove the reverse inclusion. Since position and momentum operators corresponding to different modes commute, it suffices to show that for every $f \in \bigcap_{j=1}^{d} D(P^\infty_j) \cap D(Q^\infty_j)$ and $j=1, \dots, d$, one has
$$P_jf \in D(Q_j^\infty) \text{ and } Q_jf \in D(P_j^\infty).$$

We claim that it is possible to find a sequence $(f_n)_{n \geq 0} \subset {\cal S}(\RR^d)$ such that
\begin{equation}\label{eq:approx}\lim_{n \rightarrow +\infty} P_j^kf_n=P_j^k f, \text{ and } \lim_{n \rightarrow +\infty} Q_j^kf_n=Q_j^k f, \quad \forall j=1, \dots ,d , \, k \geq 0.\end{equation}

If this is the case, in order to show that $Q_j f$ belongs to the domain of $P^k_j$ for some $k\geq 0$, we can show that $(P_j^kQ_jf_n)_{n \geq 0}$ is a Cauchy sequence (since $P_j^k$ is a closed operator). Notice that for every $g \in {\cal S}(\RR^d)$
$$\|P_j^kQ_j g\|^2=\langle g,Q_jP^{2k}_jQ_j g \rangle=\langle g,P_j^{2k}Q_j^2 g \rangle + 2ki\langle g, P_j^{2k} Q_jg \rangle \leq \|P_j^{2k}g\|\|Q^2_j g\|+ 2k\|P^{2k}_j g\|\|Q_j g\|,$$
where we used the canonical commutation relation $[Q_j,P_j]\subseteq2i\mathbf{1}.$ Therefore Eq. \eqref{eq:approx} implies that $(P^kQ_jf_n)_{n \geq 0}$ is a Cauchy sequence. Analogously one can prove that $P_j f$ belongs to the domain of $Q^k_j$. 

\bigskip Let us prove the claim. Let us consider  a mollifier $\varphi$, i.e. a smooth function with compact support such that $\int_{\RR^d}\varphi(x)dx=1$ and $\epsilon^{-d} \varphi(x/\epsilon)$ convergence to a Dirac delta in the space of tempered distributions; moreover, let us consider a smooth compactly supported function $\eta$ such that it takes values in $[0,1]$ and it is equal to $1$ in an open neighbourhood of the origin. The approximating sequence is given by $f_n=\eta_n \cdot( f *\varphi_n)$, where
$$\varphi_n(x)=n^{d}\varphi(nx), \quad \eta_n(x)=\eta(x/n).$$
Indeed, one can immediately check that $f_n \xrightarrow[n\rightarrow +\infty]{} f$ in $L^2(\RR^d)$. Moreover, for every $j=1,\dots, d$, $k \geq 1$ one has that
\[\begin{split}
    \partial_j^kf_n(x)=&\sum_{l=0}^{k}\binom{k}{l}\partial_j^l\eta_n(x)\partial_j^{k-l}(f*\varphi_n)\\
    &=\sum_{l=0}^{k}\binom{k}{l}n^{-l}\partial_j^l\eta(x/n)(\partial_j^{k-l}f)*\varphi_n\xrightarrow[n \rightarrow +\infty]{} \partial^k_j f,
\end{split}\]
where $\partial_j$ denotes the partial derivative along the $j$-th vector of the canonical basis. In the last equality we used the fact that, since $\partial_j^{k-l}f$ is well defined for all $l$s, $\partial_j^{k-l}(f*\varphi_n)=(\partial_j^{k-l}f)*\varphi_n$. Indeed,
$$\lim_{n \rightarrow +\infty}\eta(x/n)(\partial_j^{k}f)*\varphi_n=\partial_j^{k}f $$
in $L^2(\RR^{d})$ because $(\partial_j^{k}f)*\varphi_n$ converges to $\partial_j^{k}f$ in $L^2(\RR^{d})$ and for every $n \geq 1$, $\sup_{x \in \RR^{d}}\eta(x/n) =\sup_{x \in \RR^{d}}\eta(x)$. Morevoer, for $l=1, \dots, k$
$$\lim_{n \rightarrow +\infty}n^{-l}\partial^l_j\eta(x/n)(\partial_j^{k-l}f)*\varphi_n=0$$
in $L^2(\RR^{d})$, because $n^{-l}\partial^l_j\eta(x/n)$ converges to $0$ in $L^\infty(\RR^{d})$ and $\|(\partial_j^{k-l}f)*\varphi_n\|_2 \leq \|\partial^{k-l}_j f\|_2.$ Moreover,
\[\begin{split}
    \|x_j^k(f_n-f)\|_2^2\leq &\int_{\RR^d}x_j^{2k}\eta^2_n(x)\left | \int_{\RR^d}(f(x-y)-f(x))\varphi_n(y)dy\right |^2dx+ \quad (I)\\
    &\int_{\RR^d}x_j^{2k}(1-\eta_n(x))^2|f(x)|^2 dx \quad (II).
\end{split}\]
Notice that $(II)$ goes to $0$ for $n$ approaching $\infty$ by dominated convergence and
\[\begin{split}
    (I)=&\int_{\RR^d}x_j^{2k}\eta^2_n(x)\left | \int_{\RR^d}(f(x-y/n)-f(x))\varphi(y)dy\right |^2dx\leq\\
    &\int_{\RR^d}\varphi^2(y) \int_{\RR^d}x_j^{2k}\eta^2_n(x)|f(x-y/n)-f(x)|^2dxdy,
\end{split}\]
which converges to $0$ as well for $n$ going to $\infty$. Indeed, $\partial_j^k f$ is defined in $L^2(\RR^d)$ for every $k \geq 1$, then it is well known that $f$ admits a continuous representative; therefore, for every $y \in \RR^d$,
$$\lim_{n \rightarrow +\infty}\int_{\RR^d}x_j^{2k}\eta^2_n(x)|f(x-y/n)-f(x)|^2dx=0$$
by dominated convergence ($x_j^kf(x)$ and $f$ belong to $L^2(\RR^d)$. By dominated convergence again, we can conclude: indeed, $\varphi$ is compactly supported and for every $y \in \RR^d$
\[\begin{split}
    &\int_{\RR^d}x_j^{2k}\eta^2_n(x)|f(x-y/n)-f(x)|^2dx \leq  \\
    &\|x_j^k f(x)\|^2_2+\sum_{l=0}^{2k}\binom{2k}{l}|y_j/n|^{2k-l}\int_{\RR^d}|x_j-y_j/n|^{l}|f(x-y_j/n)|^2dx.
\end{split}\]
If $l$ is even
$$\int_{\RR^d}|x_j-y_j/n|^{l}|f(x-y_j/n)|^2dx=\|x_j^{l/2}f(x)\|_2^2,$$
if $l$ is odd
$$\int_{\RR^d}|x_j-y_j/n|^{l}|f(x-y_j/n)|^2dx \leq \left (\max_{\{x:x_j^{l}\geq x_j^{l+1}\}}|x_j|^{l}\right )\|f\|_2^2+\|x_j^{(l+1)/2}f(x)\|_2^2.$$
\end{proof}

\section{Proof of Theorem \ref{theo:regularisation} and Proposition \ref{prop:schwarz}} \label{app:schwarzII}

\begin{proof}[Proof of Theorem \ref{theo:regularisation}]
    Eq. \eqref{eq:regularisation} follows from the fact that if $\mathbf{w} \in \boldsymbol{\cal V}^\perp={\rm supp}(\mathbf{C}_t)$ (the equality between the linear spaces follows from Lemma \ref{lem:kerct}), then one can easily check that $\hat{\rho}_t(\mathbf{z})\in D(\widehat{{\cal H}}_{\mathbf{w}}^\infty)$ (see eq. \eqref{eq:domainHhat}) and, therefore, $\rho_t \in D({\cal H}_{\mathbf{w}}^\infty)$ thanks to the intertwining relation in eq. \eqref{eq:equivalenceofgroups}

    .

    \bigskip $1. \Leftrightarrow2.$ is an immediate consequence of Proposition \ref{prop:decofree}.

    \bigskip $2. \Leftrightarrow 3.$ follows from eq. \eqref{eq:regularisation}.

    \bigskip $4. \Rightarrow 1.$ follows from the definition of $\boldsymbol{{\cal V}}$.

    \bigskip $1. \Rightarrow 4.$ Notice that for every $t \geq 0$ and $\rho$ one has
$$\widehat{\TT_{t*}(\rho)}(\mathbf{z})=\widehat{g}_t(\mathbf{z})\cdot \widehat{\rho}(e^{t\mathbf{Z}}\mathbf{z})=\widehat{g}_t(\mathbf{z})\cdot \widehat{{\cal S}_t(\rho)}=\widehat{g_t*{\cal S}_t(\rho)},$$
where we used that the quantum Fourier transform turns convolution into product (see Proposition 3.4 in \cite{We84}) and the statement follows from injectivity of the quantum Fourier transform.
\end{proof}

\begin{proof}[Proof of Proposition \ref{prop:schwarz}]
$1.$ We will use the characterisation of ${\cal S}(\RR^{d})$ given in Lemma \ref{lem:schwarz} and we will prove by induction that $\rho_t:=\TT_{t*}(\rho)$ is such that
$$\rho_t(D(P_j^k)) \subseteq D(P_j^k), \quad \rho_t(D(Q_j^k)) \subseteq D(Q_j^k)$$
for every $j=1,\dots, d$ and $k \geq 1.$ Let us consider any $\mathbf{w} \in \{-\mathbf{f}_j/\sqrt{2}, \mathbf{f}_{2j}/\sqrt{2}\}_{j =1}^{d}$, where $\{\mathbf{f}_j\}_{j=1}^{2d}$ is the canonical basis of $\RR^{2d}$; we recall that $R(-\mathbf{f}_j/\sqrt{2})=P_j$ and $R(\mathbf{f}_{2j}/\sqrt{2})=Q_j$. Theorem \ref{theo:regularisation} ensures that $\rho_t \in D({\cal H}_{\mathbf{w}}^\infty)$; if $\varphi \in D(R(\mathbf{w})^{k})$ it is easy to see that for every $k \geq 1$, 

$${\cal H}_\mathbf{w}^k(\rho_t)\varphi=\sum_{l=0}^{k} \alpha_{lk} R(\mathbf{w})^l\rho_t R(\mathbf{w})^{k-l}\varphi$$
for some real coefficients $\alpha_{kl}.$ It follows by induction from the characterisation of $D({\cal H}_{\mathbf{w}})$ (see eq. \eqref{eq:domains}).

Using that $\rho_t \in D({\cal H}_{\mathbf{w}})$, we get that $\rho_t (D(R(\mathbf{w}))) \subseteq D(R(\mathbf{w}))$. Let us assume that $\rho_t (D(R(\mathbf{w})^k)) \subseteq D(R(\mathbf{w})^k)$ for some $k \geq 1$, let us show that this implies that $\rho_t (D(R(\mathbf{w})^{k+1})) \subseteq D(R(\mathbf{w})^{k+1})$. Indeed, by $\rho_t \in D({\cal H}_\mathbf{w}^{k+1})$, one has that
${\cal H}_\mathbf{w}^k(\rho_t)(D(R(\mathbf{w})) \subseteq D(R(\mathbf{w}))$, therefore for every $\varphi \in D(R(\mathbf{w})^{k+1})$ one has
$${\cal H}_\mathbf{w}^k(\rho_t)\varphi=\sum_{l=0}^{k} \alpha_{lk} R(\mathbf{w})^l \rho_t R(\mathbf{w})^{k-l}\varphi=\psi \in D(R(\mathbf{w})).$$
Notice that $R(\mathbf{w})^{k-l}\varphi \in D(R(\mathbf{w})^{l+1})$ and, for $l <k$, $\rho_t R(\mathbf{w})^{k-l}\varphi \in D(R(\mathbf{w})^{l+1})$. Therefore, for $l <k$, $R(\mathbf{w})^l \rho_t R(\mathbf{w})^{k-l}\varphi \in D(R(\mathbf{w}))$, hence
$$\alpha_{kk}R(\mathbf{w})^k \rho_t \varphi=\underbrace{\psi}_{\in D(R(\mathbf{w}))}-\sum_{l=0}^{k-1}\alpha_{lk}\underbrace{D(R(\mathbf{w})^l \rho_t D(R(\mathbf{w})v^{k-l}\varphi}_{\in D(R(\mathbf{w})}\in D(R(\mathbf{w})),$$
which means that $\rho_t(D(R(\mathbf{w})^{k+1}))\subseteq D(R(\mathbf{w})^{k+1})$ and we are done.

\bigskip $2.$ Due to item $4.$ in Theorem \ref{theo:regularisation}, one has that for every $t >0$ and initial state $\rho$
$${\cal H}_\mathbf{w}(\TT_{t*}(\rho))={\cal H}_\mathbf{w}(g_t*{\cal S}_t(\rho))=-\partial_{\mathbf{w}}g_t*{\cal S}_t(\rho),$$
therefore Young's inequality (Proposition 3.2 in \cite{We84}) implies that
$$\|{\cal H}_\mathbf{w}(\TT_{t*}(\rho))\|_2 \leq \|\partial_{\mathbf{w}}g_t\|_1\|{\cal S}_t(\rho)\|_2.$$
Since the quantum Fourier transform is an isometry, one has that
$$\|{\cal S}_t(\rho)\|_2=\|\widehat{\rho}(e^{t\mathbf{Z}}\mathbf{z})\|_2 =\det(e^{-t\mathbf{Z}})\|\widehat{\rho}(\mathbf{z})\|_2=e^{-t\tr(\mathbf{Z})}\|\widehat{\rho}(\mathbf{z})\|_2.$$
On the other hand, since
$$\partial_{\mathbf{w}}g_t(\mathbf{x})=\frac{\langle \mathbf{w}, \mathbf{C}^{-1}_t(\mathbf{x}-\boldsymbol{\zeta}_t)\rangle}{2\pi \det(\mathbf{C}_t)}e^{-\frac{\langle \mathbf{x}-\boldsymbol{\zeta}_t,\mathbf{C}_t^{-1}(\mathbf{x}-\boldsymbol{\zeta}_t)\rangle}{2}},$$
one has
$$\|\partial_{\mathbf{w}}g_t\|_1=\mathbb{E}[|\langle \mathbf{w},\mathbf{C}_t^{-1/2}\mathbf{X}\rangle|] \leq \|\mathbf{w}\|\|\mathbf{C}_t^{-1/2}\|\mathbb{E}[\|\mathbf{X}\|],$$
where $\mathbf{X}$ is a standard multivariate Gaussian random vector and we used a simple change of variable formula in the first equality.
\end{proof}

\section{Proof of Proposition \ref{prop:Sobolev}} \label{app:Sobo}

\begin{proof}[Proof of Proposition \ref{prop:Sobolev}] Let us consider any $x \in W^{1}$, then
$$\|(\mathbf{1}+N)^{1/2}x(\mathbf{1}+N)^{1/2}\|_2 \leq \|(\mathbf{1}+N)^{1/2}x(\mathbf{1}+N)^{1/2}\|_1<+\infty.$$
For every field operator $R(\mathbf{z})$ one has that $D(N^{1/2}) \subseteq D(R(\mathbf{z}))$, hence if $x \in W^{1}$, then $x \in D({\cal K}_{\mathbf{w}})$ for every $\mathbf{w} \in \RR^{2d}$.

By contradiction, let us assume that there exists a GQMS $\TT$ which is $1$-Sobolev regularising and let us consider an initial state $\rho \notin D({\cal K}_{\overline{\mathbf{w}}})$ for some $\overline{\mathbf{w}} \in \RR^{2d}$. The regularisation assumption made on the semigroup implies that for every $t>0$ and $\mathbf{w} \in \RR^{2d}$, one has $\rho_t:=\TT_{t*}(\rho) \in D({\cal K}_{\mathbf{w}})$; consequently,
$$\hat{\rho}_t(\mathbf{z})=e^{-\frac{\|\mathbf{C}_t\mathbf{z}\|^2}{2}+i\langle\boldsymbol{\zeta}_t, \mathbf{z}\rangle}\hat{\rho}(e^{t\mathbf{Z}}\mathbf{z}) \in D(\widehat{\cal K}_{\mathbf{w}}) \text{ for all } \mathbf{w} \in \RR^{2d}.$$
Since $e^{-\frac{\|\mathbf{C}_t\mathbf{z}\|^2}{2}+i\langle\boldsymbol{\zeta}_t, \mathbf{z}\rangle}$ is smooth and $e^{t\mathbf{Z}}\mathbf{z}$ is smooth and bijective, this implies that $\hat{\rho}(\mathbf{z})$ belongs to the domain of $\widehat{{\cal K}}_{\mathbf{w}}$ for every $\mathbf{w}$, which contradicts the hypotheses.
\end{proof}

\section{Relationship between $\LL_1$, $\LL_2$ and $\mathbf{Z}$, $\mathbf{C}_{Z}$ $\boldsymbol{\zeta}$.} \label{app:l1l2}
$\mathbf{Z}$ and $\mathbf{C}_{Z}$ represents the drift and diffusion matrices in the phase space; the choice we made of Weyl operators, i.e.
$$W(z)=e^{i\sqrt{2}\sum_{j=1}^{d}\Im(z_j)Q_j-\Re(z_j)P_j}$$
corresponds to expressing a point in the quantum phase space in coordinates with respect to:
$-\sqrt{2}P_1,
\dots
-\sqrt{2}P_d,
\sqrt{2}Q_1,
\dots
\sqrt{2}Q_d$.
However, in this proof it will be more convenient to work choosing coordinates with respect to creation and annihilation operators: since $W(z)=e^{\sum_{j=1}^{d}z_ja_j^\dagger-\overline{z_j}a_j}$, $(\Re(z_1), \dots, \Re(z_d), \Im(z_1), \dots, \Im(z_d))$ becomes $(z_1,\dots, z_d, -\overline{z_1}, \dots, \overline{z_d})$ in the ``basis" $a_1^\dagger, \dots, a_d^\dagger, a_1, \dots, a_d.$ With this choice, an operator $A$ as in \eqref{eq:A1A2} reads
$$\mathbb{A}=\begin{pmatrix} A_1 & -A_2\\
-\overline{A}_2 & \overline{A_1} \end{pmatrix}=\frac{1}{2}\begin{pmatrix} \mathbf{1} & i \mathbf{1} \\
-\mathbf{1} & i \mathbf{1}\end{pmatrix}\mathbf{A} \begin{pmatrix} \mathbf{1} & - \mathbf{1} \\
-i\mathbf{1} & -i \mathbf{1}\end{pmatrix}.$$
The right hand side of the previous equation can be extended to any linear operator on $\CC^{2d}$, even if they do not originate as complexifications of linear operators acting on $\RR^{2d}$, e.g. $\mathbf{C}_{Z}$. 
Therefore,
\begin{align}
&\mathbb{Z}=\begin{pmatrix} (U^T\overline{U}-V^T\overline{V})/2+i\Omega & (V^TU-U^T V)/2 -i \kappa\\
(V^*\overline{U}-U^*\overline{V})/2+i\overline{\kappa} & (U^*U-V^*V)/2-i\overline{\Omega}\end{pmatrix}, \label{eq:bbZ}\\
&\mathbb{C}_{Z}=\sum_{\ell=1}^{m}\mathbb{K}_{\ell}, \quad \mathbb{K}_{\ell}=\begin{pmatrix}V^{(\ell)T}\\ -U^{(\ell)*}\end{pmatrix}\begin{pmatrix}\overline{V}^{(\ell)} &-U^{(\ell)}\end{pmatrix},\label{eq:bbCZ}\end{align}
where $U^{\ell}$ and $V^{\ell}$ stay for the $\ell$-th row of $U$ and $V$, respectively.

We will show how operators in equations \eqref{eq:bbZ} and \eqref{eq:bbCZ} appear from $\LL_1$ and $\LL_2$ when one looks at their formal action on complex linear combinations of creation and annihilation operators. If one expands the commutator with respect to jump operators in terms of commutators with respect to creation and annihilation operators, one obtains
$$
\partial_{L_\ell}^*\partial_{L_\ell}=\sum_{j,k=1}^{d}(v_{\ell j}\overline{v}_{\ell k}\partial^*_{a_j}\partial_{a_k}+v_{\ell j}u_{\ell k}\partial^*_{a_j}\partial^*_{a_k}+\overline{u}_{\ell j}  \overline{v}_{\ell k}\partial_{a_j}\partial_{a_k}+\overline{u}_{\ell j} u_{\ell k} \partial_{a_j}\partial^*_{a_k}),
$$
which, in a compact way, can be rewritten as
$$(\partial^*_{a_1}, \dots, \partial^*_{a_d}, -\partial_{a_1}, \dots, -\partial_{a_d}) \mathbb{K}_{\ell}\begin{pmatrix} \partial_{a_1}\\ \dots\\\partial_{a_d}\\ -\partial^*_{a_1}\\ \dots\\ -\partial^*_{a_d}\end{pmatrix}, \quad \mathbb{K}_{\ell}=\begin{pmatrix}V^{(\ell)T}\\ -U^{(\ell)*}\end{pmatrix}\begin{pmatrix}\overline{V}^{(\ell)} &-U^{(\ell)}\end{pmatrix},$$
therefore
$${\cal L}_2=-\frac{1}{2}(\partial^*_{a_1}, \dots, \partial^*_{a_d}, -\partial_{a_1}, \dots, -\partial_{a_d}) \mathbb{K}\begin{pmatrix} \partial_{a_1}\\ \dots\\\partial_{a_d}\\ -\partial^*_{a_1}\\ \dots\\ -\partial^*_{a_d}\end{pmatrix}.$$
Notice that $\partial_{a_1}, \dots, \partial_{a_d},-\partial_{a_1^\dagger}, \dots, -\partial_{a_d^\dagger}$ can be seen as a ``dual basis" with respect to $a_1^\dagger, \dots, a_d^\dagger, a_1, \dots, a_d,$ in that $\partial_{a_j}(a_k)=\partial_{a_j^\dagger}(a_k^\dagger)=0$ and $\partial_{a_j}(a^\dagger_k)=-\partial_{a_j^\dagger}(a_k)=\delta_{jk}\mathbf{1}.$ Concerning $\LL_1$, one can see that $i\partial_H$ encodes all the information regarding
$$\mathbb{H}=\begin{pmatrix} i\Omega & -i\kappa \\
i\overline{\kappa} & -i \overline{\Omega}\end{pmatrix} \text{ and } \boldsymbol{\zeta}.$$
Indeed, one has
\[\begin{split}i\partial_H&=i\sum_{k,j=1}^{d} \left ( \Omega_{jk} \partial_{a^\dagger_j a_k} + \frac{\kappa_{jk}}{2}\partial_{a_j^\dagger a_k^\dagger}+\frac{\overline{\kappa}_{jk}}{2}\partial_{a_j a_k}\right )+\frac{1}{2}\sum_{j=1}^{d} \zeta_j \partial^*_{a_j} + \overline{\zeta}_j \partial_{a_j}\\
&=i\sum_{k,j=1}^{d} \left (\Omega_{jk} (L_{a_j}^*\partial_{a_k}+R_{a_k}\partial_{a_j}^*)\right .\\
&\left. + \frac{\kappa_{jk}}{2}(L_{a_j}^*\partial^*_{a_k}+R^*_{a_k}\partial_{a_j}^*)+\frac{\overline{\kappa}_{jk}}{2}(L_{a_j}\partial_{a_k}+R_{a_k}\partial_{a_j}) \right )\\
&+\frac{1}{2}\sum_{j=1}^{d} \zeta_j \partial^*_{a_j} + \overline{\zeta}_j \partial_{a_j}\\
&=(L_{a_1}^*, \dots, L_{a_d}^*)i\Omega\begin{pmatrix} \partial_{a_1}\\ \dots\\
\partial_{a_d}\end{pmatrix}+(R_{a_1}, \dots, R_{a_d})(-i\overline{\Omega})\begin{pmatrix} -\partial^*_{a_1}\\ \dots\\
-\partial^*_{a_d}\end{pmatrix}\\
&+(L_{a_1}^*, \dots, L_{a_d}^*)\left (\frac{-i\kappa}{2}\right )\begin{pmatrix} -\partial^*_{a_1}\\ \dots\\
-\partial^*_{a_d}\end{pmatrix}+(R_{a_1}^*, \dots, R_{a_d}^*)\left (\frac{-i\kappa}{2}\right )\begin{pmatrix} -\partial^*_{a_1}\\ \dots\\
-\partial^*_{a_d}\end{pmatrix}
\end{split}\]
\[\begin{split}
&+(L_{a_1}, \dots, L_{a_d})\left (\frac{i\overline{\kappa}}{2}\right )\begin{pmatrix} \partial_{a_1}\\ \dots\\
\partial_{a_d}\end{pmatrix}+(R_{a_1}, \dots, R_{a_d})\left (\frac{i\overline{\kappa}}{2}\right )\begin{pmatrix} \partial_{a_1}\\ \dots\\
\partial_{a_d}\end{pmatrix}\\
&+\frac{1}{2}\sum_{j=1}^{d} \zeta_j \partial^*_{a_j} + \overline{\zeta}_j \partial_{a_j}.\end{split}\]

Moreover, ${\cal L}_1(\cdot)-i\partial_H$ can be expanded as:
\[\begin{split}&R_{L_\ell}\partial_{L_\ell}^*-R_{L^*_\ell}\partial_{L_\ell}=\\
&\sum_{j,k=1}^{d}(v_{\ell j}\overline{v}_{\ell k} R_{a_k}\partial_{a_j}^*+v_{\ell j}u_{\ell k} R_{a^*_k}\partial_{a_j}^*+R_{a_k}\overline{u}_{\ell j}\overline{v}_{\ell k} \partial_{a_j}+\overline{u}_{\ell j}u_{\ell k} R_{a^*_k}\partial_{a_j}+\\
&-u_{\ell j}\overline{u}_{\ell k} R_{a_k}\partial_{a_j}^*-u_{\ell j}v_{\ell k} R_{a^*_k}\partial_{a_j}^*-\overline{v}_{\ell j}\overline{u}_{\ell k} R_{a_k}\partial_{a_j}-\overline{v}_{\ell j}v_{\ell k} R_{a^*_k}\partial_{a_j}.\end{split}\]
This can be written in a more compact way as
$$(R^*_{a_1},\dots,R^*_{a_d},R_{a_1}, \dots, R_{a_d}) \widetilde{\mathbb{Z}}\begin{pmatrix} \partial_{a_1}\\ \dots\\ \partial_{a_d}\\-\partial^*_{a_1}\\ \dots\\ -\partial^*_{a_d}\end{pmatrix}, \quad \widetilde{\mathbb{Z}}:=\frac{1}{2}\begin{pmatrix} U^T\overline{U}-V^T\overline{V} & V^TU-U^T V\\
V^*\overline{U}-U^*\overline{V} & U^*U-V^*V\end{pmatrix}.$$
Since $\mathbb{Z}=\mathbb{H}+\widetilde{\mathbb{Z}}$, we just showed that $\LL_1$ and $\mathbb{Z}$, $\boldsymbol{\zeta}$ encode the same information.

\section{Proof of Theorem \ref{th:equivcond}} \label{app:C}

\begin{proof}[Proof of Theorem \ref{th:equivcond}]
\textbf{\eqref{eq:Gcond} $\Leftrightarrow$ \eqref{eq:CZcond}.} First of all, notice that \eqref{eq:CZcond} is equivalent to
\begin{center}
    there are no non-trivial $\mathbb{Z}$-invariant subspaces in the kernel of $\mathbb{K}$;
\end{center}
this is because the matrices involved are equivalent with respect to the same change of basis. It is convenient to translate linear algebraic statements in terms of superoperators: we noticed in \ref{app:C} that $\mathbb{Z}$ describes the action of the superoperator
$$\widetilde{{\cal L}}_1={\cal L}_1-\frac{i}{2}\partial_{R(\mathbf{J}\boldsymbol{\zeta})},$$
when formally applied to complex linear combinations of creation and annihilation operators in the basis $a^\dagger_1,\dots, a^\dagger_d, a_1, \dots, a_d$. Moreover, we also observed that the kernel of $\mathbb{K}$ corresponds to
$$B_0:=\bigcap_{\ell=1, \dots, m}\ker(\partial_{L_\ell}).$$

In the same way, $\mathbb{G}$ corresponds to the superoperator $\partial_{\widetilde{G}},$ where
$$\widetilde{G}:=G+\frac{i}{2}R(\mathbf{J}\boldsymbol{\zeta}).
$$
In order to keep the notation simple, without loss of generality, we can assume that $\boldsymbol{\zeta}=0$ and keep the notation without wiggles. \eqref{eq:Gcond} can be reformulated as
    \begin{equation} \label{eq:commutantG}
    B_\infty:=\bigcap_{k \geq 0}B_k=\{0\}, \text{ where } B_k:=\bigcap_{\ell=1,\dots, m}\ker(\partial_{\partial^k_G(L_\ell)}).\end{equation}
Notice that for every $k\geq 0$
\begin{equation} \label{eq:partialG}
\partial_{\partial^{k+1}_G(L_\ell)}=\partial_G\partial_{\partial^k_G(L_\ell)}-\partial_{\partial^k_G(L_\ell)}\partial_G,
\end{equation}
therefore condition \eqref{eq:commutantG} can be rephrased as
\begin{equation} \label{cond:Ginvariance}
    \text{there exists no non-trivial $\partial_G$-invariant subspaces in }B_0=\bigcap_{\ell=1, \dots, m}\ker(\partial_{L_\ell}).
\end{equation}
Indeed, by definition $B_\infty$ is contained in $B_0=\bigcap_{\ell=1, \dots, m}\ker(\partial_{L_\ell})$; let us show that it is also $\partial_G$-invariant: if $v \in B_\infty$, then using eq. \eqref{eq:partialG} one has that for every $k \geq 1$
$$\partial_{\partial^k_G(L_\ell)}\partial_G(v)=\partial_G\partial_{\partial^k_G(L_\ell)}(v)-\partial_{\partial^{k+1}_G(L_\ell)}(v)=0$$
and we are done.

Let us now prove the reverse implication: it is enough to prove that any $\partial_G$-invariant subspace in $B_0$ is contained in $B_\infty$. Let us consider a $\partial_G$-invariant linear space in $B_0$ that we denote by $A$. It is enough to show, that if $A \subseteq B_k$, then $A\subseteq B_{k+1}$: let us consider $v \in A$, then
$\partial_{\partial^{k+1}_G(L_\ell)}(v)=\partial_G\partial_{\partial^k_G(L_\ell)}(v)-\partial_{\partial^k_G(L_\ell)}\partial_G(v)=0,$
where we used that $\partial_G(A) \subseteq A \subseteq B_k$.

In order to conclude, we only need to show that $\partial_G$ and $\LL_1$ share the same invariant subspaces in $B_0$. Notice, that restricted to $B_0$, one has
\[\begin{split}
    \partial_G&=-i \partial_H-\frac{1}{2}\sum_{\ell=1}^{m}(L_{L_\ell}^*\partial_{L_\ell}+R_{L_\ell}\partial^*_{L_\ell})\\
    &=-i\partial_H-\frac{1}{2}\sum_{\ell=1}^{m}(R_{L_\ell}\partial_{L_\ell}^*-R_{L_\ell}^*\partial_{L_\ell})=-{\cal L}_1,
\end{split}\]
where we used the fact that $\partial_{L_\ell}$ is identically zero on $B_0$.

\bigskip 
$\eqref{eq:CZcond} \Leftrightarrow \eqref{eq:Horcond}.$ It is just a tedious computation that we report below for completeness.  One can check that the following commutation rules hold true:
    \begin{itemize}
    \item $[\partial_{a_j}, \partial_{a_k}]=[\partial^*_{a_j}, \partial_{a_k}]=[\partial^*_{a_j}, \partial^*_{a_k}]=0;$
    \item $[\partial_{a_j}, R_{a_k}]=[\partial^*_{a_j}, R^*_{a_k}]=[\partial_{a_j}, L_{a_k}]=[\partial^*_{a_j}, L^*_{a_k}]=0;$
    \item $[R^*_{a_k},\partial_{a_j}]=[L^*_{a_k},\partial_{a_j} ]=-\delta_{jk}{\rm Id}=-[R_{a_k},\partial^*_{a_j}]=-[L_{a_k},\partial^*_{a_j}];$
    \item $[\partial_{a_j^*a_k}, \partial_{a_m}]=-\delta_{jm}\partial_{a_k},$ $[\partial_{a_j^*a_k}, \partial^*_{a_m}]=\delta_{km}\partial_{a_j}^*$;
    \item $[\partial_{a_j^*a^*_k}, \partial_{a_m}]=-\delta_{jm}\partial^*_{a_k}-\delta_{km}\partial_{a_j}^*,$ $[\partial_{a_ja_k}, \partial^*_{a_m}]=\delta_{jm}\partial_{a_k}+\delta_{km}\partial_{a_j}.$
    \end{itemize}
    Let us compute $[{\cal L}_1,\sum_{m=1}^{d}\overline{v}_m\partial_{a_m}+u_m\partial_{a_m}^*].$ We will break the computations in several steps: one has
    \[\begin{split}
        &\left [i\sum_{j,k=1}^{d}\Omega_{jk} \partial_{a_j^*a_k},\sum_{m=1}^{d}\overline{v}_m\partial_{a_m}+u_m\partial_{a_m}^*\right ]=\\
        &i\sum_{j,k,m=1}^{d}\Omega_{jk}(\overline{v}_m[\partial_{a_j^*a_k},\partial_{a_m}]+u_m[\partial_{a_j^*a_k},\partial^*_{a_m}])=\\
        &-i\sum_{j,k=1}^{d}\Omega_{jk}\overline{v}_j\partial_{a_k}+i\sum_{j,k=1}^{d}\Omega_{jk}u_k\partial^*_{a_j},
    \end{split}\]
    hence the matrix expression of $ \left [i\sum_{j,k=1}^{d}\Omega_{jk} \partial_{a_j^*a_k},\cdot \right ]$ in the basis $\{ \partial_{a_1}, \dots, \partial_{a_d}, -\partial^*_{a_1}, \dots, -\partial_{a_d}^*\}$ is given by
    $$\begin{pmatrix} -i\overline{\Omega} & 0\\
    0 & i\Omega
        \end{pmatrix},$$
        where we used that $\Omega$ is Hermitian. Let us move to the rest of the Hamiltonian:
     \[\begin{split}
        &\left [\frac{i}{2}\sum_{j,k=1}^{d}\kappa_{jk} \partial_{a_j^*a^*_k}+\overline{\kappa}_{jk} \partial_{a_ja_k},\sum_{m=1}^{d}\overline{v}_m\partial_{a_m}+u_m\partial_{a_m}^*\right ]=\\
        &\frac{i}{2}\sum_{j,k,m=1}^{d}\kappa_{jk}(\overline{v}_m[\partial_{a_j^*a^*_k},\partial_{a_m}]+u_m[\partial_{a_j^*a^*_k},\partial^*_{a_m}])+\\
        &\frac{i}{2}\sum_{j,k,m=1}^{d}\overline{\kappa}_{jk}(\overline{v}_m[\partial_{a_ja_k},\partial_{a_m}]+u_m[\partial_{a_ja_k},\partial^*_{a_m}])=\\
        &-\frac{i}{2}\sum_{j,k}^{d}\kappa_{jk}(\overline{v}_j\partial_{a_k}^*+\overline{v}_k\partial^*_{a_j})+\\
        &\frac{i}{2}\sum_{j,k=1}^{d}\overline{\kappa}_{jk}(u_j\partial_{a_k}+u_k \partial_{a_j})
    \end{split}\]
    and the matrix expression of $ \left [\frac{i}{2}\sum_{j,k=1}^{d}\kappa_{jk} \partial_{a_j^*a^*_k}+\overline{\kappa}_{jk} \partial_{a_ja_k},\cdot \right ]$ is given by
    $$\begin{pmatrix} 0 & -i\overline{\kappa}\\
    i\kappa & 0
        \end{pmatrix},$$
        where we used that $\kappa$ is symmetric. One can immediately see that the term with $\zeta$ brings zero contribution. Therefore, the Hamiltonian part is represented by
        $$\begin{pmatrix} -i\overline{\Omega} & -i\overline{\kappa}\\
    i\kappa & i\Omega
        \end{pmatrix}=-\mathbb{H}^T.$$
    Let us now look at the remaining part:
    \begin{equation}\label{eq:ztildecomm}\begin{split}
        &\left [ \sum_{j,k=1}^{d} \widetilde{Z}^{11}_{jk}R^*_{a_j}\partial_{a_k}-\widetilde{Z}^{12}_{jk}R^*_{a_j}\partial^*_{a_k}+\right .\\
        &\left .\widetilde{Z}^{21}_{jk}R_{a_j}\partial_{a_k}-\widetilde{Z}^{22}_{jk}R_{a_j}\partial^*_{a_k},\sum_{m=1}^{d}\overline{v}_m\partial_{a_m}+u_m\partial_{a_m}^*\right ]=\\
        &\sum_{j,k,m=1}^{d} \widetilde{Z}^{11}_{jk}(\overline{v}_m[R_{a_j}^*\partial_{a_k},\partial_{a_m}]+u_m[R_{a_j}^*\partial_{a_k},\partial^*_{a_m}])+\\
        &\sum_{j,k,m=1}^{d} -\widetilde{Z}^{12}_{jk}(\overline{v}_m[R_{a_j}^*\partial^*_{a_k},\partial_{a_m}]+u_m[R_{a_j}^*\partial^*_{a_k},\partial^*_{a_m}])+\\
        &\sum_{j,k,m=1}^{d} \widetilde{Z}^{21}_{jk}(\overline{v}_m[R_{a_j}\partial_{a_k},\partial_{a_m}]+u_m[R_{a_j}\partial_{a_k},\partial^*_{a_m}])+\\
        &\sum_{j,k,m=1}^{d} -\widetilde{Z}^{22}_{jk}(\overline{v}_m[R_{a_j}\partial^*_{a_k},\partial_{a_m}]+u_m[R_{a_j}\partial^*_{a_k},\partial^*_{a_m}]).
    \end{split}\end{equation}
    Let us simplify the commutators appearing in the expression above using the following set of equations:
   \begin{align*}
       &[R_{a_j}^*\partial_{a_k},\partial_{a_m}]=[R_{a_j}^*,\partial_{a_m}]\partial_{a_k}=-\delta_{jm}\partial_{a_k}, \quad [R_{a_j}^*\partial_{a_k},\partial^*_{a_m}]=0,\\
    &[R_{a_j}^*\partial^*_{a_k},\partial_{a_m}]=[R_{a_j}^*,\partial_{a_m}]\partial^*_{a_k}=-\delta_{jm}\partial^*_{a_k}, \quad [R_{a_j}^*\partial^*_{a_k},\partial^*_{a_m}]=0,\\
    &[R_{a_j}\partial_{a_k},\partial_{a_m}]=0, \quad [R_{a_j}\partial_{a_k},\partial^*_{a_m}]=[R_{a_j},\partial^*_{a_m}]\partial_{a_k}=\delta_{jm}\partial_{a_k},\\
    &[R_{a_j}\partial^*_{a_k},\partial_{a_m}]=0, \quad [R_{a_j}\partial^*_{a_k},\partial^*_{a_m}]=[R_{a_j},\partial^*_{a_m}]\partial^*_{a_k}=\delta_{jm}\partial^*_{a_k}.
    \end{align*}
    Therefore the expression in eq. \eqref{eq:ztildecomm} becomes
    \[
     \sum_{j,k,1}^{d} -\widetilde{Z}^{11}_{jk}\overline{v}_j\partial_{a_k}+\sum_{j,k=1}^{d} +\widetilde{Z}^{12}_{jk}\overline{v}_j\partial^*_{a_k}+\sum_{j,k=1}^{d} \widetilde{Z}^{21}_{jk}u_j \partial_{a_k}+\sum_{j,k=1}^{d}- \widetilde{Z}^{22}_{jk}u_j\partial^*_{a_k}.
    \]

The matrix expression of this last part of the drift term is given by $-\mathbb{Z}^T.$ 

Summing up, expressing everything with respect to the basis $\partial_{a_1}, \dots, \partial_{a_d}, -\partial_{a_1}^*, \dots, -\partial_{a_d}^*$ one has that
$$\partial_{L_\ell} \Leftrightarrow B^{(\ell)}:=\begin{pmatrix}{\overline{V}^{(\ell)T}}\\ -U^{(\ell)T} \end{pmatrix}, \quad  [{\cal L}_1,\partial_{L_\ell}]\Leftrightarrow-\mathbb{Z}^TB^{(\ell)}, \quad [{\cal L}_1,[{\cal L}_1,\partial_{L_\ell}]]\Leftrightarrow(-\mathbb{Z})^{2T}B^{(\ell)}\dots $$

We deduced that \eqref{eq:Horcond} is equivalent to
$${\rm span}_{\CC}\{-\mathbb{Z}^{kT}B^{(\ell)}: k \geq 0\}=\CC^{2d},$$
which can be reformulated as
$${\rm span}_{\CC}\{B^{(\ell)T}\mathbb{Z}^{k}: k \geq 0\}=\CC^{2d}$$
and the latter form can easily be recast as \eqref{eq:CZcond}.
\end{proof}

\section{Proof of Theorem \ref{theo:main}} \label{app:maintheo}

\begin{proof}[Proof of Theorem \ref{theo:main}]
$1. \Rightarrow 2.$ By contradiction, let us assume that $\TT$ is not irreducible, then there exists a nontrivial subharmonic projection $p$; we recall that $p$ reduces $\TT$ in the sense that for any initial state $\rho$ such that ${\rm supp}(\rho) \subseteq {\rm supp}(p)$, and for every $t \geq 0$ one has
\begin{equation} \label{eq:reducing}{\rm supp}(\TT_{t*}(\rho)) \subseteq {\rm supp}(p).\end{equation}

There exists a nonzero Schwarz function $\varphi \in {\cal S}(\RR^d)$ such that $\ket{\psi}:=\TT_{t*}(\rho)\ket{\varphi} \neq 0$; otherwise, thanks to the density of Schwarz functions in $L^2(\RR^d)$, this would imply that $\TT_{t*}(\rho)=0$, which is false. Using Proposition \ref{prop:schwarz} and eq. \eqref{eq:reducing}, one gets that $\ket{\psi}$ belongs to ${\cal S}(\RR^d)\cap {\rm supp}(p)$. Let us consider the closed linear subspace $V$ generated by $\ket{\psi}$ and vectors of the form $q(a_j,a_k^\dagger) \ket{\psi},$ where $q(a_j,a_k^\dagger)$ is any monomial in $a_j$ and $a_k^\dagger$; we remark that it is well defined since $\ket{\psi} \in {\cal S}(\RR^d)$ and, consequently. Moreover, let us define $W$ as the closed linear space generated by 
\begin{equation} \label{eq:commutant}\ket{\psi}, \, \partial_G^{m_1}(L_{\ell_1}) \cdots \partial_G^{m_n}(L_{\ell_n})\ket{\psi} \quad n\geq 1,\, m_1, \dots, m_n \geq 0.
\end{equation}
It is immediate to see that $W \subseteq V$, since $\partial_G^{m_1}(L_{\ell_1}) \cdots \partial_G^{m_n}(L_{\ell_n})$ is a polynomial in $a_j$ and $a_k^\dagger$ with complex coefficients (this can be seen by expressing $L_\ell$s as linear combinations of creation and annihilation operators). However, the reverse inclusion is true as well, since we can express creation and annihilation operators as linear combinations of the identity operator and $\partial_G^{k}(L_\ell)$ for $\ell=1,\dots, m$ and $k \geq 0$ thanks to $\eqref{eq:Gcond}.$ Therefore, $V=W$. Now, one can easily see that $W \subseteq {\rm supp}(p)$: indeed,
$$P_{t_1}L_{\ell_1}P_{t_2-t_1}\cdots P_{t_{n}-t_{n-1}}L_{\ell_n}P_{t_n}\ket{\psi} \in {\rm supp}(p)$$
thanks to Theorem \ref{theo:subharmonic}. Then so does
$$\partial^{m_1}_{t_1}\cdots \partial_{t_n}^{m_n}P_{t_1}L_{\ell_1}P_{t_2-t_1}\cdots P_{t_{n}-t_{n-1}}L_{\ell_n}P_{t_n}\ket{\psi}|_{t_1=\cdots =t_n=0},$$
which is exactly equal to the term appearing in eq. \eqref{eq:commutant}. Finally, we can use the same idea as in the proof of Theorem 4 in \cite{AFP22b} in order to prove that $V=\hh\subseteq {\rm supp}(p) \subsetneq\hh$, getting to a contradiction.

\bigskip $2. \Rightarrow 1.$ Let us assume that there exists a nontrivial linear space $\boldsymbol{{\cal X}}$ of $\CC^{2d}$ which is $\mathbf{Z}$-invariant in $\ker(\mathbf{C}_{Z})$; we will show that this implies the existence of a nontrivial subharmonic projection for $\TT$. There are three possible cases:
\begin{enumerate}
\item $\mathbf{Z}$ admits a real eigenvector in $\boldsymbol{{\cal X}}$, i.e. there exists $\mathbf{v} \in \boldsymbol{{\cal X}}\cap \RR^{2d}$ such that $\mathbf{Z}\mathbf{v}=\alpha \mathbf{v}$ for some $\mathbf{v}$ and $\alpha \in \RR$. In this case, one has
$$\TT_{t}(W(u\mathbf{v}))=e^{iu\int_{0}^te^{\alpha s}ds\langle \boldsymbol{\zeta},\mathbf{v} \rangle }W(e^{\alpha t}u\mathbf{v}), \quad u \in \mathbb{R},$$
which means that if we denote by $\chi_{A}(R(\mathbf{v}))$ the spectral projection of $R(\mathbf{v})$ corresponding to the Borel set $A$, then
$$\TT_t(\chi_{A}(R(\mathbf{v}))=\chi_{e^{-\alpha t}A-\int_0^te^{-\alpha(t-s)}ds\langle \boldsymbol{\zeta},\mathbf{v} \rangle}(R(\mathbf{v})).$$
One can easily find nontrivial subharmonic projections of the form $\chi_{A}(R(\mathbf{v}))$ for every value of $\alpha$ and $\beta:=\langle \boldsymbol{\zeta},\mathbf{v} \rangle$:
\begin{itemize}
\item $\alpha<0$, one can pick $A=\{x \in \RR: \, |x+\beta/\alpha|<\epsilon\}$ for every $\epsilon>0$;
\item $\alpha=0$, one can pick $A=\{x \in \RR:\, x\beta  \geq 0\},$
\item $\alpha>0$, one can pick $A=\{x \in \RR: \, |x+\beta/\alpha|>\epsilon\}$ for every $\epsilon>0$.
\end{itemize}
The computations justifying this claim can be found for the more general two dimensional case in item 3. below.
\item $\mathbf{Z}$ does not admit any real eigenvector in $\boldsymbol{{\cal X}}$, but there exists a complex eigenvector $\mathbf{z} \in \boldsymbol{{\cal X}}$ such that $\mathbf{Z} \mathbf{z}=(\alpha+i\beta)\mathbf{z}$ and, if we define $\mathbf{v}:=\Re(\mathbf{z})$ and $\mathbf{w}=\Im(\mathbf{z})$, one has that $\langle \mathbf{v}, \mathbf{J}\mathbf{w} \rangle\neq 0$. Let us define ${\cal W}:={\rm span}_{\RR}\{\mathbf{v}, \mathbf{w}\}$. In this case, $\TT$ maps the Von Neumann algebra $\{W(\mathbf{s}): \, \mathbf{s} \in {\cal W}\}^{\prime\prime}$ into itself and its action on it is unitarily equivalent to that of a one-mode GQMS whose quantum diffusion matrix has a non-trivial kernel. In this case, Theorem 7 in \cite{AFP21} and Theorem \ref{th:equivcond} together ensure that there exists a nontrivial subharmonic projection $p \in \{W(\mathbf{s}): \, \mathbf{s} \in {\cal W}\}^{\prime\prime}$.
\item Finally, let us consider the case in which $\mathbf{Z}$ does not admit any real eigenvector in $\boldsymbol{{\cal X}}$ and for every complex eigenvector $\mathbf{z} \in \boldsymbol{{\cal X}}$, i.e. $\mathbf{Z} \mathbf{z}=(\alpha+i\beta)\mathbf{z}$, one has $\langle \mathbf{v}, \mathbf{J}\mathbf{w} \rangle= 0$, where $\mathbf{v}:=\Re(\mathbf{z})$ and $\mathbf{w}=\Im(\mathbf{z})$. In this case $R(\mathbf{v})$ and $R(\mathbf{w})$ admit a common spectral resolution and we denote the spectral projections by $\chi_{A}(R(\mathbf{v}),R(\mathbf{w})) $ for every Borel set $A \subseteq \RR^{2}$. Moreover, notice that
$$
0=\langle \mathbf{z},\mathbf{C}_{Z}\mathbf{z} \rangle=\langle \mathbf{z},\mathbf{C}\mathbf{z}\rangle -i\langle \mathbf{z},(\mathbf{Z}^T\mathbf{J}+\mathbf{J}\mathbf{Z})\mathbf{z} \rangle=\langle \mathbf{z},\mathbf{C}\mathbf{z}\rangle,$$
where we used that $\mathbf{z}=\mathbf{v}+i\mathbf{w}$, that $\boldsymbol{{\cal W}}:={\rm span}_{\RR}\{\mathbf{v}, \mathbf{w}\}$ is $\mathbf{Z}$-invariant and that for every $\mathbf{s}, \mathbf{r}\in \boldsymbol{{\cal W}}$ one has $\langle \mathbf{s}, \mathbf{J} \mathbf{r} \rangle=0$. Now, we can expand further and write
$$0=\langle \mathbf{z},\mathbf{C}\mathbf{z}\rangle=\langle \mathbf{v}, \mathbf{C}\mathbf{v} \rangle + \langle \mathbf{w}, \mathbf{C}\mathbf{w} \rangle \Leftrightarrow 0=\langle \mathbf{v}, \mathbf{C}\mathbf{v} \rangle = \langle \mathbf{w}, \mathbf{C}\mathbf{w} \rangle, $$
therefore we obtained that $\boldsymbol{{\cal W}} \subseteq \ker(\mathbf{C})$ and one has
$$\TT_t(W(\mathbf{s}))=e^{i \int_0^t\langle \boldsymbol{\zeta},e^{s\mathbf{Z}}\mathbf{s} \rangle ds}W(e^{t \mathbf{Z}}\mathbf{s}), \quad \mathbf{s} \in \boldsymbol{{\cal W}}, \, t \geq 0,$$
which means that the action on the commutative algebra generated by $\{\chi_{A}(R(\mathbf{v}), R(\mathbf{w})): A \text{ Borel subset of } \RR^2\}$ is unitarily equivalent to the one on $L^2(\RR^2)$ acting on characteristic functions in the following way:
$$\chi_A(x) \mapsto \chi_{A_t}(x), \quad A_t:=e^{-B^Tt}A-\int_0^te^{-B^T(t-s)}bds,$$
where $B=\begin{pmatrix} \alpha & -\beta\\
\beta & \alpha \end{pmatrix}$ and $b$ is some vector in $\RR^2.$
Therefore, the projections in $\{W(\mathbf{s}):\, \mathbf{s} \in \boldsymbol{{\cal W}}\}^{\prime\prime}$ corresponding to the following characteristic functions are nontrivial subharmonic projections:
\begin{itemize}
    \item if $\alpha<0$ or $\alpha=0$ and $\beta \neq 0$, one can pick the characteristic function corresponding to $A:=\{x \in \RR^2:\, \|x+B^{-T}b\|<\epsilon\}$;
    \item if $\alpha=\beta=0$, one can pick the characteristic function corresponding to $A:=\{x \in \RR^2:\, \langle x,b\rangle \geq 0\}$;
    \item if $\alpha >0$, one can pick the characteristic function corresponding to $A:=\{x \in \RR^2:\, \|x+B^{-T}b\|>\epsilon\}$.
\end{itemize}
Indeed, this can be easily checked using the fact that
$$x \in A_t \Leftrightarrow e^{tB^T}x+\int_0^t e^{sB^T}b ds=\begin{cases} x+tb & \text{if }\alpha=\beta=0,\\e^{tB^T}(x+B^{-T}b)-B^{-T}b & \text{otherwise}\end{cases}$$
(notice that $B^T$ is invertible unless $\alpha=\beta=0$) and observing that for every $x^\prime \in \RR^2$, $t \geq 0$
$$\|e^{tB^{T}}x^\prime\|\begin{cases} \leq \|x^\prime\| & \text{if } \alpha \leq 0,\\
>\|x^\prime\| & \text{otherwise.}\end{cases}$$
\end{enumerate}
\end{proof}

\bibliographystyle{abbrv}
\bibliography{biblio.bib}

\begin{thebibliography}{10}

\bibitem{AFP21}
J.~Agredo, F.~Fagnola, and D.~Poletti.
\newblock Gaussian quantum {M}arkov semigroups on a one-mode {F}ock space:
  irreducibility and normal invariant states.
\newblock {\em Open Syst. Inf. Dyn.}, 28(1):Paper No. 2150001, 39, 2021.

\bibitem{AFP22}
J.~Agredo, F.~Fagnola, and D.~Poletti.
\newblock The decoherence-free subalgebra of {G}aussian quantum {M}arkov
  semigroups.
\newblock {\em Milan Journal Of Mathematics}, 90(1):257--289, 2022.

\bibitem{AFP22b}
J.~Agredo, F.~Fagnola, and D.~Poletti.
\newblock The {K}ossakowski matrix and strict positivity of {M}arkovian quantum
  dynamics.
\newblock {\em Open Syst. Inf. Dyn.}, 29(2):Paper No. 2250005, 16, 2022.

\bibitem{AF08}
A.~Arnold, F.~Fagnola, and L.~Neumann.
\newblock Quantum {F}okker-{P}lanck models: the {L}indblad and {W}igner
  approaches.
\newblock In {\em Quantum probability and related topics}, volume~23 of {\em
  QP--PQ: Quantum Probab. White Noise Anal.}, pages 23--48. World Sci. Publ.,
  Hackensack, NJ, 2008.

\bibitem{BSFQ26}
J.~R. Bola\~{n}os Servin, F.~Fagnola, and R.~Quezada.
\newblock Gaussian quantum markov semigroups of weak coupling limit type.
\newblock {\em Infinite Dimensional Analysis, Quantum Probability and Related
  Topics}, 0(0):2550012, 0.

\bibitem{CJ20}
R.~Carbone and A.~Jen\v~cov\'a.
\newblock On period, cycles and fixed points of a quantum channel.
\newblock {\em Ann. Henri Poincar\'e}, 21(1):155--188, 2020.

\bibitem{CS08}
R.~Carbone and E.~Sasso.
\newblock Hypercontractivity for a quantum {O}rnstein-{U}hlenbeck semigroup.
\newblock {\em Probab. Theory Related Fields}, 140(3-4):505--522, 2008.

\bibitem{CFL00}
F.~Cipriani, F.~Fagnola, and J.~M. Lindsay.
\newblock Spectral analysis and {F}eller property for quantum
  {O}rnstein-{U}hlenbeck semigroups.
\newblock {\em Comm. Math. Phys.}, 210(1):85--105, 2000.

\bibitem{Co82}
M.~Courbage.
\newblock Mathematical problems of irreversible statistical mechanics for
  quantum systems. {I}. {A}nalytic continuation of the collision and
  destruction operators by spectral.
\newblock {\em J. Math. Phys.}, 23(4):646--651, 1982.

\bibitem{DFSU16}
J.~Deschamps, F.~Fagnola, E.~Sasso, and V.~Umanit\`{a}.
\newblock Structure of uniformly continuous quantum markov semigroups.
\newblock {\em Reviews in Mathematical Physics}, 28(01):1650003, 2016.

\bibitem{FG25}
F.~Fagnola and F.~Girotti.
\newblock Irreducibility of quantum {M}arkov semigroups, uniqueness of
  invariant states and related properties.
\newblock {\em Complex Anal. Oper. Theory}, 20(5):Paper No. 121, 30, 2026.

\bibitem{FL25}
F.~Fagnola and Z.~Li.
\newblock Spectral analysis for gaussian quantum markov semigroups, 2025.

\bibitem{FP22}
F.~Fagnola and D.~Poletti.
\newblock On irreducibility of {G}aussian quantum {M}arkov semigroups.
\newblock {\em Infin. Dimens. Anal. Quantum Probab. Relat. Top.}, 25(4):Paper
  No. 2240001, 19, 2022.

\bibitem{FPSU24}
F.~Fagnola, D.~Poletti, E.~Sasso, and V.~Umanità.
\newblock The spectral gap of a gaussian quantum markovian generator.
\newblock {\em Journal of Functional Analysis}, 289(10):111119, 2025.

\bibitem{FR02}
F.~Fagnola and R.~Rebolledo.
\newblock Subharmonic projections for a quantum {M}arkov semigroup.
\newblock {\em J. Math. Phys.}, 43(2):1074--1082, 2002.

\bibitem{GP26}
F.~Girotti and D.~Poletti.
\newblock Gaussian quantum {M}arkov semigroups on finitely many modes admitting
  a normal invariant state.
\newblock {\em J. Math. Anal. Appl.}, 556(1):Paper No. 130150, 37, 2026.

\bibitem{GMR24}
P.~Gondolf, T.~M{\"{o}}bus, and C.~Rouz{\'{e}}.
\newblock Energy preserving evolutions over {B}osonic systems.
\newblock {\em {Quantum}}, 8:1551, Dec. 2024.

\bibitem{KKW16}
M.~Keyl, J.~Kiukas, and R.~F. Werner.
\newblock Schwartz operators.
\newblock {\em Rev. Math. Phys.}, 28(3):1630001, 60, 2016.

\bibitem{Ho82}
A.~S. A.~S. Kholevo.
\newblock {\em Probabilistic and statistical aspects of quantum theory / A.S.
  Holevo.}
\newblock North-Holland series in statistics and probability ; v. 1.
  North-Holland, Amsterdam ; Oxford, 1982.

\bibitem{KP04}
C.~K. Ko and Y.~M. Park.
\newblock Construction of a family of quantum {O}rnstein-{U}hlenbeck
  semigroups.
\newblock {\em J. Math. Phys.}, 45(2):609--627, 2004.

\bibitem{Li25}
Z.~Li.
\newblock On the existence of the kms spectral gap in gaussian quantum markov
  semigroups.
\newblock {\em arXiv preprint arXiv:2512.23414}, 2025.

\bibitem{LHB24}
D.~Lonigro, A.~Hahn, and D.~Burgarth.
\newblock On the liouville–von neumann equation for unbounded hamiltonians.
\newblock {\em Open Systems \& Information Dynamics}, 31(04):2450018, 2024.

\bibitem{LMP20}
A.~Lunardi, G.~Metafune, and D.~Pallara.
\newblock The {O}rnstein-{U}hlenbeck semigroup in finite dimension.
\newblock {\em Philos. Trans. Roy. Soc. A}, 378(2185):20200217, 15, 2020.

\bibitem{MTG09}
S.~Meyn, R.~L. Tweedie, and P.~W. Glynn.
\newblock {\em Markov Chains and Stochastic Stability}.
\newblock Cambridge Mathematical Library. Cambridge University Press, 2
  edition, 2009.

\bibitem{Nu13}
D.~Nualart.
\newblock {\em The Malliavin Calculus and Related Topics}.
\newblock Probability and Its Applications. Springer New York, NY, 1 edition,
  2013.

\bibitem{NY17}
H.~I. Nurdin and N.~Yamamoto.
\newblock {\em Linear Dynamical Quantum Systems: Analysis, Synthesis, and
  Control}.
\newblock Springer Publishing Company, Incorporated, 1st edition, 2017.

\bibitem{PO22}
D.~Poletti.
\newblock Characterization of gaussian quantum markov semigroups.
\newblock {\em Infinite Dimensional Analysis, Quantum Probability and Related
  Topics}, 25(03):2250014, 2022.

\bibitem{PSU26}
D.~Poletti, E.~Sasso, and V.~Umanit\`a.
\newblock Symmetric {G}aussian quantum {M}arkov semigroups.
\newblock {\em J. Math. Phys.}, 67(2):Paper No. 022102, 2026.

\bibitem{RGM25}
P.~C. Rico, P.~Gondolf, and T.~Möbus.
\newblock Instantaneous sobolev regularization for dissipative bosonic
  dynamics, 2025.

\bibitem{SP23}
R.~Schubert and T.~Plastow.
\newblock Decoherence time scales and the h\"ormander condition, 2023.

\bibitem{SXZ26}
L.~Sun, Z.~Xu, and H.~Zhang.
\newblock Hypercontractivity for a family of quantum ornstein-uhlenbeck
  semigroups, 2026.

\bibitem{We84}
R.~Werner.
\newblock Quantum harmonic analysis on phase space.
\newblock {\em J. Math. Phys.}, 25(5):1404--1411, 1984.

\bibitem{Ya14}
N.~Yamamoto.
\newblock Decoherence-free linear quantum subsystems.
\newblock {\em IEEE Transactions on Automatic Control}, 59(7):1845--1857, 2014.

\bibitem{DZ22}
G.~Zhang and Z.~Dong.
\newblock Linear quantum systems: A tutorial.
\newblock {\em Annual Reviews in Control}, 54:274--294, 2022.

\end{thebibliography}

\end{document}